\documentclass[9pt,journal]{IEEEtran}

\usepackage{amsmath}
\usepackage{amssymb}
\usepackage{graphicx}
\usepackage{amsthm}
\usepackage{mathtools}
\usepackage{bm}
\usepackage{xcolor}
\usepackage{hyperref}
 \usepackage{subeqnarray}
 \usepackage{array}

\usepackage{threeparttable}
\usepackage{array}
\newcolumntype{P}[1]{>{\centering\arraybackslash}p{#1}}
 
\usepackage[switch]{lineno}
\usepackage[named]{algo}
\usepackage[noend]{algpseudocode}
\usepackage{algorithm}

\usepackage{cite}

\newtheorem{theorem}{Theorem}

\newtheorem{lemma}[theorem]{Lemma}
\newtheorem{proposition}[theorem]{Proposition}

\theoremstyle{remark}

\ifCLASSINFOpdf

\else

\fi

\begin{document}

\title{CoSTAP: Clutter Suppression in Co-Pulsing FDA-STAP}

\author{Wanghan Lv, \IEEEmembership{Member, IEEE} and Kumar Vijay Mishra, \IEEEmembership{Senior Member, IEEE}
\thanks{This work was supported by the National Natural Science Foundation of China under Grant 62001116 and the Natural Science Foundation of the Jiangsu Higher Education Institutions of China under Grant 23KJB510008. The conference precursor of this work was presented at the 2022 Asilomar Conference on Signals, Systems, and Computers (ACSSC) \cite{WLv2022_1}.}
\thanks{W. L. is with the College of Computer and Information Engineering, Nanjing Tech University, Nanjing 211816, China, e-mail: lwanghan@njtech.edu.cn. }
\thanks{K. V. M. is with The University of Maryland, College Park, MD 20742 USA, e-mail: mishra@umd.edu.}
}

\maketitle

\IEEEpeerreviewmaketitle

\begin{abstract}
Range-dependent clutter suppression poses significant challenges in airborne frequency diverse array (FDA) radar, where resolving range ambiguity is particularly difficult. Traditional space-time adaptive processing (STAP) techniques used for clutter mitigation in FDA radars operate in the physical domain defined by first-order statistics. In this paper, unlike conventional airborne uniform FDA, we introduce a space-time-range adaptive processing (STRAP) method to exploit second-order statistics for clutter suppression in the newly proposed co-pulsing FDA radar. This approach utilizes \textit{c}o-prime frequency offsets (FOs) across the elements of a \textit{c}o-prime array, with each element transmitting at a non-uniform \textit{c}o-prime pulse repetition interval (C-Cube). By incorporating second-order statistics from the co-array domain, the co-pulsing STRAP or \textit{CoSTAP} benefits from increased degrees of freedom (DoFs) and low computational cost while maintaining strong clutter suppression capabilities. However, this approach also introduces significant computational burdens in the coarray domain. To address this, we propose an approximate method for three-dimensional (3-D) clutter subspace estimation using discrete prolate spheroidal sequences (DPSS) to balance clutter suppression performance and computational cost. We first develop a 3-D clutter rank evaluation criterion to exploit the geometry of 3-D clutter in a general scenario. Following this, we present a clutter subspace rejection method to mitigate the effects of 
interference such as jammer. Compared to existing FDA-STAP algorithms, our proposed CoSTAP method offers superior clutter suppression performance, lower computational complexity, and enhanced robustness to interference. Numerical experiments validate the effectiveness and advantages of our method.
\end{abstract}


\section{Introduction}
Space-time adaptive processing (STAP) has garnered significant attention in military applications such as airborne reconnaissance, early warning, and surveillance \cite{JWard1995,BTang2016}. This technique jointly processes spatial samples from an array of antenna elements with the fast- and slow-time temporal samples of target echoes. Extensive studies have demonstrated that STAP effectively removes co-existing clutter and other types of interference (such as jamming, clustered Dopplers, range-angle spreading), thereby enhancing detection of low-Doppler targets \cite{WLMelvin2004,CChen2007}. However, modern military applications increasingly employ high-speed platforms, where the radar transmits waveforms at high pulse repetition frequency (PRF) to mitigate Doppler ambiguity. But this also leads to severe range ambiguity problems \cite{MXing2011}, wherein near- and far-range clutter coexist in the same range cell. As a result, existing STAP approaches fail to align the aliasing clutter thereby significantly degrading clutter suppression performance.

Recently, a flexible beam scanning array known as a frequency diverse array (FDA) has been proposed. It features range-angle dependent beampatterns that can suppress range-ambiguous clutter and improve detection performance \cite{PAntonik2006,PBaizert2006}. By exploiting the extra degrees of freedom (DoFs) in the range domain, FDA radar outperforms phased-array (PA) radar in parameter estimation \cite{JXu2015}, mainlobe jammer mitigation \cite{JXu2015_R2}, and range-ambiguous clutter suppression \cite{JXu2015_1}. An adaptive range-angle-Doppler processing approach was proposed in \cite{JXu2017} for airborne FDA-MIMO radar systems with severe range ambiguity problems, where the maximum resolvable number of range ambiguities was derived in closed form.

Considering the importance of clutter rank in distinguishing clutter subspace from noise subspace, clutter rank evaluation is crucial for reducing computational complexity in STAP. Several clutter rank evaluation rules have been developed for uniform linear arrays (ULA) in FDA radars \cite{YYan2020,KWang2022}. Recently, co-prime FDA structures have been proposed to provide enhanced performance in parameter estimation without changes in size, weight, power consumption, or cost \cite{PPVaidyanathan2011,MWang2017,CZhou2018}. These FDA structures exploit the coarray concept to achieve higher DoFs \cite{SQin2017,ZMao2022}. Furthermore, \cite{WLv2022} introduced a co-pulsing FDA radar that offers significantly large DoFs for target localization in the range-azimuth-elevation-Doppler domain. In this paper, we focus on \textit{CoSTAP}, i.e., STAP for co-pulsing FDA radar.

While several STAP algorithms have been developed for clutter suppression using co-prime arrays \cite{XWang2018,XWang2020}, these methods cannot be directly applied to co-prime FDA in the presence of range-ambiguous clutter because the clutter spectrum in the classical co-prime STAP model is range independent. Moreover, while the co-prime scheme enhances aperture and increases DoFs in the virtual domain, it also leads to a rapid increase in computational burden and the number of required training samples, which are often deficient in practice. To balance increased DoFs with computational complexity, we propose exploiting the well-known prolate spheroidal wave function (PSWF) \cite{}. 

Foundational theory of PSWFs appeared in a series of seminal papers by Slepian, Pollack, and Landau \cite{DSlepian1961prolate,HJLandau1962,DSlepian1978}. In particular, discrete prolate spheroidal sequences (DPSSs) \cite{DSlepian1978} are the discrete versions of PSWFs or Slepian sequences. The DPSSs are a collection of orthogonal bandlimited sequences that are most concentrated in time to a given index range and yield a highly efficient basis for sampled bandlimited functions. These DPSS characteristics have lent their applications in spectral estimation \cite{APapoulis1975,MHayes1983} and wireless communications channel estimation \cite{TZemen2005}. For correlated MIMO channel representation, DPSS vectors are selected as a suitable predefined basis to avoid the frequency leakage effect of Fourier basis expansion \cite{OLongoria2011}. Further, modulated DPSS frames have been proposed for fast fading channels to preserve sparsity and enhance estimation accuracy \cite{ESejdic2008}.

In radar signal processing, there is a rich heritage of research on employing DPSS for clutter suppression. In \cite{CYChen2008}, DPSSs were used to approximate the clutter subspace in airborne MIMO radar, significantly reducing the computational complexity compared to traditional STAP methods. To further reduce the complexity of computing the signal representation using DPSS, \cite{ZZHU2018} proposed rapid orthogonal approximate Slepian transform (ROAST) method based on fast Slepian transform (FST). The computational burden of ROAST is comparable to fast Fourier transform (FFT). Utilizing FST, \cite{LOsadciw2021} proposed an enhanced STAP method through Slepian transform and achieved better clutter cancellation than sample matrix inversion (SMI) or eigendecomposition STAP. Besides STAP, other related radar applications such as subspace radar resolution cell rejection \cite{JBosse2018} and multipath suppression \cite{BPDay2020} have also demonstrated effective use of DPSS.

Unlike existing DPSS-based clutter suppression methods, we focus on the application of DPSS in CoSTAP, proposing a two-stage DPSS-based clutter suppression method. We first introduce a novel space-time-range adaptive processing (STRAP) filter for clutter suppression using co-pulsing FDA radar, which offers the advantages of larger aperture and increased DoFs. We evaluate and analyze the clutter rank of airborne co-pulsing FDA radar for different number of range ambiguities in a general case. To mitigate the computational burden, we propose an algorithm for approximating the clutter subspace of FDA using PSWF. Note that, unlike the FDA-MIMO STAP method in \cite{JXu2017}, our aforementioned CoSTAP does not employ a MIMO array.

Preliminary results of this work appeared in our conference publication \cite{WLv2022_1}, which introduced the conceptual ideas behind CoSTAP. In this paper, we present a comprehensive description of DPSS-based CoSTAP, including a more generalized closed form of three-dimensional (3-D) clutter rank, theoretical guarantees, and comprehensive numerical evaluation. Our main contributions are:\\
\textbf{1) DPSS-based clutter subspace approximation for co-pulsing FDA.} Our previous work \cite{WLv2022} introduced co-pulsing FDA radar with the goal of reducing resources such as sensors, spectrum, and dwell time. In this paper, we consider co-pulsing FDA radar on an airborne platform for clutter suppression and propose a STRAP method in the coarray domain, namely CoSTAP technique. With the merit of increased number of virtual array elements and pulses in coarray domain, we achieve better performance in terms of clutter characterization capability and output signal-to-interference-plus-noise ratio (SINR) compared with the traditional uniform FDA-STAP in the physical domain. However, the rank of 3-D clutter subspace arising out of joint space-time-range dimension in the 3-D coarray domain is dramatically high and, therefore, affects both the complexity and the convergence of the 3-D STAP algorithm. To make a trade-off between performance and complexity, we propose an approximate method of 3-D clutter subspace estimation in co-pulsing FDA system by leveraging upon DPSS functions thereby gaining a computational cost advantage. Prior works on reduced-rate measurements using DPSS showed recovery of under-sampled multiband signals using multiband modulated DPSS dictionary. However, DPSSs remain relatively unexamined for coarray radar applications.\\
\textbf{2) Generalization of clutter rank evaluation criterion.} Clutter rank is a critical evaluation criterion for analyzing clutter subspace. Traditional STAP literature indicates that clutter subspace has a small rank \cite{JWard1995,WFeng2017,ACombernoux2019}, which makes clutter rank estimation in the reduced-dimension (RD) domain essential for classical STAP algorithms. In FDA radar, clutter rank analysis has typically been confined to the spatial domain, specifically the joint transmit-receive dimension \cite{KWang2022}. This work extends the analysis to the 3-D correlation domain by incorporating the pulse dimension. Unlike \cite{KWang2022}, which requires the spacing ratio between transmit and receive arrays (represented by the ratio $\beta$ between platform Doppler frequency and PRF) to be an integer, our generalized closed form of the 3-D clutter rank is derived without any constraints on $\beta$, based on co-prime number theory.\\
\textbf{3) Mitigation of additional clustered interference.} Classical STAP only suppresses white-noise jammer interference, which is spread across the entire frequency region and originates from a specific angle \cite{WLMelvin2004}. However, STAP struggles with barrage interference that have clustered Dopplers, ranges, and angles, which are present across large regions and overlap with targets \cite{QLiu2020,SCheng2023}. These strong undesired interference sources can mask the desired targets. To address this issue, we explore interference subspace approximation using DPSS to reject these types of interference from the covariance matrix of the received signal. Additionally, we derive a modulated DPSS construction method under the 3-D Kronecker structure scenario, which incorporates the clutter information in space-time-range coarray domain. 

The rest of the paper is organized as follows. In the next section, we introduce the signal model of airborne co-prime FDA radar in the physical domain. Section \ref{sec:3-D_Adaptive} converts the signal model from physical domain to correlation domain and proposes a 3-D CoSTAP algorithm. In Section \ref{sec:DPSS}, we explore the clutter subspace and its rank in the correlation domain. Using DPSS, we construct a data-independent basis for clutter signals. Meanwhile, to enhance the CoSTAP performance, an additional subspace rejection method is exploited to eliminate all clustered components in an undesired area. In Section \ref{sec:NE}, we demonstrate the performance of our proposed method through numerical examples. We conclude in Section \ref{sec:Conclusions}.

\emph{Notations} Throughout this paper, boldface lowercase letters, such as $\mathbf{x}$, stand for vectors, and boldface uppercase letters, such as $\mathbf{X}$, stand for matrices. $\mathbf{x}[n]$ represents the $n$-th element of the vector $\mathbf{x}$. We denote the transpose, conjugate, and Hermitian by $(\cdot)^T$, $(\cdot)^*$ and $(\cdot)^H$, respectively; $\otimes$, $\odot$, and $\diamond$ represent the Kronecker, Khatri-Rao, and Hadamard products, respectively; and $\mathrm{vec}(\cdot)$ is the vectorization operator that turns a matrix into a vector by stacking all columns on top of the another. The notation $\lfloor \cdot \rfloor$ indicates the greatest integer smaller than or equal to the argument; $\mathrm{E}[\cdot]$ is the statistical expectation function; $\forall x$ means that for all $x$.  $\mathrm{gcd}(M,N)$ returns the greatest common divisor of the integers $M$ and $N$. $\langle\cdot, \cdot\rangle$ stands for the inner product.

\section{Signal Model}               \label{sec:SM}
Consider an airborne side-looking co-prime FDA transceiver that employs a $P_s$-sensor array and $K$-pulse train in a coherent processing interval (CPI). The frequency offsets (FOs) of the transmitted frequency diverse signal and the sensor spacings share the same co-prime pattern, which is characterized by the co-prime integers set $\mathcal{S}_s= \{M_s i,0\leq i \leq N_s-1 \} \cup \{N_s j,1\leq j \leq 2M_s-1 \} $ with the cardinality $\vert \mathcal{S}_s\vert = P_s= N_s+2M_s-1$. $M_s$ and $N_s$ are a pair of co-prime integers with $M_s < N_s$. As such, the set of sensor positions and all carrier frequencies are denoted by $\bm{\xi}=\mathcal{S}_s d$ and $\bm{\mathcal{F}}=f_b+ \mathcal{S}_s \Delta f $, respectively, where $d$ is regarded as the unit inter-element spacing, $f_b$ is the reference frequency, and $\Delta f$ is the unit FO. Similarly, the pulse intervals take the co-prime pattern characterized by the co-prime integers set $\mathcal{S}_t= \{M_t i,0\leq i \leq N_t-1 \} \cup \{N_t j,1\leq j \leq 2M_t-1 \} $ with $\vert \mathcal{S}_t\vert = K= N_t+2M_t-1$, yielding pulse starting time instants $\bm{\eta}= \mathcal{S}_tT$ where $T$ is the unit pulse repetition interval (PRI) for the case of uniform pulsing. Thus, the $m$-th, $0\leq m \leq P_s-1$, transmitted frequency diverse signal is
\begin{align}
s_{m}(t)=A_{m} e^{\mathrm{j}2\pi f_{m}t},~~ 0\leq t\leq T_p,
\end{align}
where $A_m$ is the pulse amplitude, $f_m$ is the $m$-th element of the set $\bm{\mathcal{F}}$, and $T_p$ is the pulse duration. Without loss of generality, set $A_m=1$.

\begin{figure}[t]
\centerline{\includegraphics[scale=0.88]{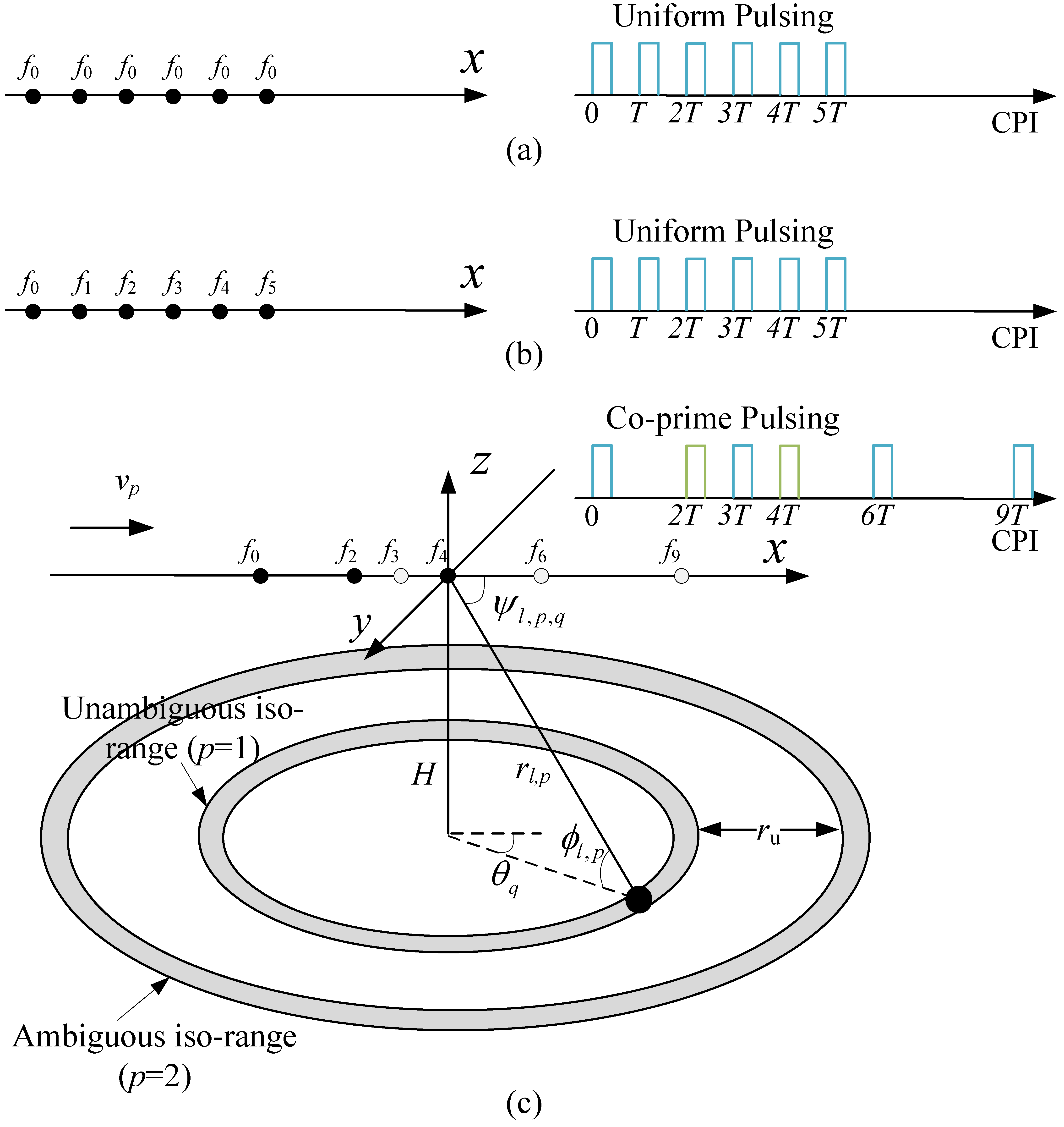}}
\caption{Airborne side-looking array radar geometry with $P_s=6$ elements and transmission sequence with $K=6$ pulses is illustrated for (a) standard STAP, (b) uniform FDA-STAP, and (c) co-pulsing FDA-STAP with $M_s=2$ and $N_s=3$.}
\label{Airborn_CoFDA}
\end{figure}

Fig. \ref{Airborn_CoFDA} shows the geometry of our proposed airborne radar system. The velocity and height of the platform are $\nu_{\mathrm{p}}$ and $H$, respectively. The slant range of the $l$-th range bin in the $p$-th ambiguous range region is 
\begin{align}
r_{l,p} = r_l+ (p-1)r_u,
\end{align}
where $r_l$ is the slant range corresponding to the $l$-th range bin in the principal range region, $p$ is the index of ambiguous range region and $r_u=cT/2$ is the maximum unambiguous range. For an arbitrary clutter scatterer at slant range $r_{l,p}$, azimuth angle $\theta_q$ and elevation angle $\phi_{l,p}$, the two-way time delay of the signal received by the $n$-th, $0 \leq n \leq P_s - 1$ sensor at the $k$-th, $0 \leq k \leq K - 1$, pulse is
\begin{align}
\tau_{n,k} &= \frac{1}{c}\left( 2r_{l,p}- 2\nu_{\mathrm{p}} \eta_k T \cos \psi_{l,p,q} -\xi_n d\cos\psi_{l,p,q}  \right),
\end{align}
where $\psi_{l,p,q}$ is the conic angle of the clutter scatterer satisfying $\cos\psi_{l,p,q}=\cos\theta_q \cos\phi_{l,p}$. $\eta_k$ is the $k$-th element of the set $\bm{\eta}$ and $\xi_n$ is the $n$-th element of the set $\bm{\xi}$. To ensure an array structure deprived of angle ambiguity, the cosine values of conic angles of different cluttter scatterers are assumed to be distinct, namely, 
\begin{align}   \label{eq:SM_R3}
    \cos \psi_{l,p,q} \neq \cos \psi_{l',p',q'},~ \forall~ (l,p,q) \neq (l',p',q').
\end{align}
Then, the received signal of the $k$-the pulse at the $n$-th sensor is 
\begin{align}     
 x_{n,k}(t,r_{l,p},\psi_{l,p,q}) = \sum_{m=0}^{P_s-1}\rho_{l,p,q} e^{\mathrm{j}2\pi f_{m}(t-\tau_{n,k})},
\end{align}
where $\rho_{l,p,q}$ is the complex reflectivity of the point scatterer. 

After demodulation and band-pass filtering, the signals corresponding to the respective frequencies $\{f_m\}_{m=0}^{P_s-1}$ are obtained. Sampling the $k$-th echo at the rate $1/T_p$ in fast-time $t_s = \eta_k T + l T_p$, where $l = 0, 1,\ldots,L - 1 = \lfloor T/T_p \rfloor$, yields the baseband received signal as
\begin{align}   \label{eq:SM_R5}
&x_{k,n,m}(r_{l,p},\psi_{l,p,q}) \nonumber\\
&= \rho_{l,p,q}  e^{ \frac{-\mathrm{j} 4\pi f_m r_{l,p}}{c} }    e^{\mathrm{j}2\pi \frac{ d \cos\psi_{l,p,q}}{\lambda_m}\xi_n}  e^{\mathrm{j} \frac{4\pi v_{\mathrm{p}}\eta_k T \cos\psi_{l,p,q}}{\lambda_m} }    \nonumber \\
&\approx \widetilde{\rho}_{l,p,q}  e^{\frac{-\mathrm{j}4\pi \xi_m \Delta f r_{l,p}}{c}}   e^{\mathrm{j}2\pi \frac{ d \cos\psi_{l,p,q}}{\lambda_b}\xi_n} e^{\frac{\mathrm{j}4\pi \nu_{\mathrm{p}} \eta_k T\cos \psi_{l,p,q}}{\lambda_b}},
\end{align}
where $\lambda_m=c/f_m$, $c$ is the speed of light, and $\widetilde{\rho}_{l,p,q}= \rho_{l,p,q}e^{ \frac{-\mathrm{j} 4\pi f_b r_{l,p}}{c} }$. Our focus is the DoF enhancement at the receive array. Thus, for simplicity, assume that the direction-of-departure (DoD) information about transmit array is embodied in $\rho_{l,p,q}$. The conic angle $\psi_{l,p,q}$ corresponds to the direction-of-arrival (DoA) of the scatterer \cite{SQin2017}.

Define $f_T(r_{l,p}) \triangleq -2\Delta f r_{l,p}/c $, $f_R(\psi_{l,p,q}) \triangleq \frac{d}{\lambda_b}\cos\psi_{l,p,q}$ and $f_d(\psi_{l,p,q}) \triangleq 2\nu_{\mathrm{p}} T \cos \psi_{l,p,q}/{\lambda_b}$. Then, stacking all baseband samples along transmit spatial, receive spatial, and temporal
dimensions, we obtain the vector form of the signal (\ref{eq:SM_R5}) as
\begin{subequations}
\begin{align}             \label{eq:SM6}
 \mathbf{x}(r_{l,p},\psi_{l,p,q}) \!=\! \widetilde{\rho}_{l,p,q} \mathbf{v}(f_T(r_{l,p}),f_d(\psi_{l,p,q}),f_R(\psi_{l,p,q})), 
\end{align}
where 
\begin{align} 
& \mathbf{v}(f_T(r_{l,p}),f_d(\psi_{l,p,q}),f_R(\psi_{l,p,q}))    \nonumber \\
&= \mathbf{a}_T(f_T(r_{l,p}))\otimes \mathbf{b}(f_d(\psi_{l,p,q})) \otimes \mathbf{a}_R(f_R(\psi_{l,p,q})), \\
&\mathbf{a}_T(f_T(r_{l,p}))          \nonumber  \\
&= [1,e^{\mathrm{j}2\pi f_T(r_{l,p})\xi_1},\ldots,e^{j2\pi f_T(r_{l,p})\xi_{P_s-1}}]^T  \in \mathbb{C}^{P_s \times 1}, \label{eq:SM7}\\
&\mathbf{a}_R(f_R(\psi_{l,p,q}))  \nonumber\\
&= [1,e^{\mathrm{j}2\pi f_R(\psi_{l,p,q})\xi_1},\ldots,e^{\mathrm{j}2\pi f_R(\psi_{l,p,q})\xi_{P_s-1}}]^T   \in \mathbb{C}^{P_s \times 1},\\
&\mathbf{b}(f_d(\psi_{l,p,q}))  \nonumber\\
&= [1,e^{\mathrm{j}2\pi f_d(\psi_{l,p,q})\eta_1},\ldots,e^{\mathrm{j}2\pi f_d(\psi_{l,p,q})\eta_{K-1}}]^T     \in   \mathbb{C}^{K \times 1}.
\end{align}
\end{subequations}

Then, the clutter return corresponding to the $l$-th iso-range bin including range ambiguity with noise is 
\begin{align}            \label{eq:SM10}
\mathbf{y}_l = \mathbf{c}_l + \mathbf{n}_l= \sum_{p=1}^{N_p} \sum_{q=1}^{N_c} \mathbf{x}(r_{l,p},\psi_{l,p,q}) +\mathbf{n}_l,
\end{align}
where $\mathbf{c}_l$ is the clutter component, $\mathbf{n}_l$ is the noise component, $N_p$ is the number of range ambiguities and $N_c$ is the number of statistically independent clutter scatterers within the iso-range bin. Considering that the clutter component $\mathbf{c}_l$ varies with slant range, we use secondary range dependence compensation (SRDC) technique \cite{JXu2015} to obtain the independent and identically distributed (i.i.d.) training samples as
\begin{align}
    \widetilde{\mathbf{y}}_l= \sum_{p=1}^{N_p} \sum_{q=1}^{N_c} \widetilde{\rho}_{l,p,q} \mathbf{v}(\tilde{f}_T(p),f_d(\psi_{l,p,q}),f_R(\psi_{l,p,q}))  + \mathbf{n}_l,
\end{align}
where $\mathbf{v}(\tilde{f}_T(p),f_d(\psi_{l,p,q}),f_R(\psi_{l,p,q}))$ is in terms of compensated transmit steering vector $\mathbf{a}_T(\tilde{f}_T(p)) $, receive steering vector $\mathbf{a}_R(f_R(\psi_{l,p,q}))$, and time steering vector $\mathbf{b}(f_d(\psi_{l,p,q}))$ with $\tilde{f}_T(p)=-2\Delta f r_u (p-1)/c$. 

Assume the target range, conic angle, and radial velocity are $r_0, \psi_0$, and $\nu_0$, respectively. The presumed steering vector of target is
\begin{align}
\mathbf{v}_t(p_0,\psi_0,\nu_0) = \mathbf{a}_T(\tilde{f}_T(p_0)) \otimes \mathbf{b}(f_d(\psi_0))\otimes \mathbf{a}_R(f_R(\psi_0)),
\end{align}
where $p_0$ is the ambiguous index of target. The STAP filter weight vector in the physical domain which maximizes the output SINR is \cite{JXu2017}
\begin{align}              \label{eq:PA3}
\mathbf{w} = \frac{\mathbf{R}^{-1} \mathbf{v}_t(p_0,\psi_0,\nu_0)}{\mathbf{v}_t^H(p_0,\psi_0,\nu_0) \mathbf{R}^{-1}\mathbf{v}_t(p_0,\psi_0,\nu_0)},
\end{align}
where $\mathbf{R}$ is the clutter covariance matrix, computed as
\begin{align}     \label{eq:PA2}
 \mathbf{R}    &= \mathrm{E}[\widetilde{\mathbf{y}}_l \widetilde{\mathbf{y}}_l^H]     \notag \\
&= \sum_{p=1}^{N_p} \sum_{q=1}^{N_c} \sigma^2_{p,q}  \mathbf{v}(\tilde{f}_T(p),f_d(\psi_{p,q}),f_R(\psi_{p,q})) \nonumber\\
&~~~~~~~ \times \mathbf{v}^H(\tilde{f}_T(p),f_d(\psi_{p,q}),f_R(\psi_{p,q})) + \sigma_n^2 \mathbf{I}
\end{align}
with $\sigma^2_{p,q}=\mathrm{E}[\vert\widetilde{\rho}_{l,p,q}\vert^2]$ and $\sigma_n^2$ being the power of clutter patch and noise, respectively. In practice, the sample covariance matrix is computed as
\begin{align}       \label{eq:R_PA3}
  \widehat{\mathbf{R}}= \sum_{l=1}^L \widetilde{\mathbf{y}}_l \widetilde{\mathbf{y}}_l^H /L. 
\end{align}

In the physical domain, the weights in (\ref{eq:PA3}) are similar to those used in the context of STRAP for FDA-MIMO \cite{JXu2015_R2}, but differ in the characteristics of space-time-range steering vectors. In the FDA-MIMO case, the direction of departure (DoD) information about the transmit array is exploited. When FDA serves as the transmit array, the enhancement in degrees of freedom (DoF) comes from the sum coarray generation based on classical MIMO techniques \cite{JLi2007} and additional controllable DoFs in the range domain. Since the DoD is equal to the direction of arrival (DoA) in a colocated MIMO scenario and range information is coupled with the DoD, the transmit and receive spatial frequencies are linearly related in the transmit-receive spatial domains, with a shift corresponding to the index of the range region \cite{JXu2017}. In this work, we aim to exploit frequency diversity technique for airborne co-prime array \cite{XWang2020}, where the focus is at the receive end. As a result, for the sake of simplicity, we ignore the effect of waveform diversity and DoD information from transmit array. Thus, with the absence of DoD, the transmit and receive spatial frequencies are independent with each other in the transmit-receive spatial domains. In the next section, we show that the DoFs are enhanced by constructing a difference coarray in space-time-range domain.

\section{Adaptive Range-Angle-Doppler Processing in the correlation domain}                       \label{sec:3-D_Adaptive}
The weight $\mathbf{w}$ in (\ref{eq:PA3}) is designed in the physical domain whose DoFs are restricted by the number of physical elements and pulses. We now propose a 3-D CoSTAP algorithm, which takes advantage of large aperture, low mutual coupling, and increased DoFs in the virtual domain provided by co-prime scheme \cite{WLv2022}, to suppress the clutter via space-time-range processing.

\subsection{Virtual space-time-range snapshot construction}
Using the identity $(\mathbf{A}\otimes \mathbf{B} \otimes \mathbf{C})(\mathbf{D}\otimes \mathbf{E} \otimes \mathbf{F})= (\mathbf{A}\mathbf{D}) \otimes (\mathbf{B}\mathbf{E}) \otimes (\mathbf{C}\mathbf{F})$ \cite{RAHorn1985}, rewrite the outer product of the space-time-range steering vectors in (\ref{eq:PA2}) as
\begin{align}                \label{eq:VS1}
&\mathbf{v}(\tilde{f}_T(p),f_d(\psi_{p,q}),f_R(\psi_{p,q})) \mathbf{v}^H(\tilde{f}_T(p),f_d(\psi_{p,q}),f_R(\psi_{p,q}))                    \notag \\
&=[\mathbf{a}_T(\tilde{f}_T(p)) \otimes \mathbf{b}(f_d(\psi_{p,q}))\otimes \mathbf{a}_R(f_R(\psi_{p,q}))] \nonumber\\ 
&~~~~~ \times [\mathbf{a}_T(\tilde{f}_T(p)) \otimes \mathbf{b}(f_d(\psi_{p,q}))\otimes \mathbf{a}_R(f_R(\psi_{p,q}))]^H      \notag   \\
&=[\mathbf{a}_T(\tilde{f}_T(p))\mathbf{a}_T^H(\tilde{f}_T(p))]\otimes [\mathbf{b}(f_d(\psi_{p,q}))\mathbf{b}^H(f_d(\psi_{p,q}))]\nonumber\\ &~~~~~ \otimes [\mathbf{a}_R(f_R(\psi_{p,q}))\mathbf{a}_R^H(f_R(\psi_{p,q}))].
\end{align}
The entries of $\mathbf{a}_T(\tilde{f}_T(p))\mathbf{a}_T^H(\tilde{f}_T(p))$ take the form of $e^{\mathrm{j}2\pi \tilde{f}_T(p)(\xi_m-\xi_n)}, 0\leq m,n\leq P_s-1$. Following \cite{EBouDaher2015}, the contiguous elements from $e^{-\mathrm{j}2\pi \tilde{f}_T(p)L_s}$ to $e^{\mathrm{j}2\pi \tilde{f}_T(p)L_s}$ are obtained, where $L_s \triangleq M_sN_s+M_s-1$. Similarly, the contiguous elements from $e^{-\mathrm{j}2\pi f_R(\psi_{p,q})L_s}$ to $e^{\mathrm{j}2\pi f_R(\psi_{p,q})L_s}$ are obtained from $\mathbf{a}_R(f_R(\psi_{p,q}))\mathbf{a}_R^H(f_R(\psi_{p,q}))$ and the contiguous elements from $e^{-\mathrm{j}2\pi f_d(\psi_{p,q})L_t}$ to $e^{\mathrm{j}2\pi f_d(\psi_{p,q})L_t}$ can be obtained from $\mathbf{b}(f_d(\psi_{p,q}))\mathbf{b}^H(f_d(\psi_{p,q}))$ where $L_t \triangleq M_tN_t+M_t-1$. Thus, by vectorizing (\ref{eq:PA2}) and picking out the contiguous entries, we can construct the virtual space-time-range snapshot as
\begin{align}                   \label{eq:VS2}
\mathbf{z}= \sum_{p=1}^{N_p} \sum_{q=1}^{N_c} \sigma^{2}_{p,q}  \breve{\mathbf{v}}(\tilde{f}_T(p),f_d(\psi_{p,q}),f_R(\psi_{p,q}))  + \sigma^2_n \mathbf{e},
\end{align}
where $\breve{\mathbf{v}}(\tilde{f}_T(p),f_d(\psi_{p,q}),f_R(\psi_{p,q})) = \breve{\mathbf{a}}_T(\tilde{f}_T(p)) \otimes \breve{\mathbf{b}}(f_d(\psi_{p,q}))\otimes  \breve{\mathbf{a}}_R(f_R(\psi_{p,q}))$ is the virtual steering vector with $\breve{\mathbf{a}}_T(\tilde{f}_T(p)) \in \mathbb{C}^{(2L_s+1)\times 1}$, $\breve{\mathbf{a}}_R(\tilde{f}_R(\psi_{p,q})) \in \mathbb{C}^{(2L_s+1)\times 1}$ and $\breve{\mathbf{b}}(f_d(\psi_{p,q})) \in \mathbb{C}^{(2L_t+1)\times 1}$. $\mathbf{e} = \mathbf{e}_T \otimes \mathbf{e}_d \otimes \mathbf{e}_R\in \mathbb{C}^{(2L_s+1)^2(2L_t+1)\times 1}$ is a vector of all zeros except a 1 at the $\frac{(2L_s+1)^2(2L_t+1)+1}{2}$-th position, $\mathbf{e}_T, \mathbf{e}_R \in \mathbb{C}^{(2L_s+1)\times 1}$ are both vectors of all zeros except a 1 at the $(L_s+1)$-th position, $\mathbf{e}_d \in \mathbb{C}^{(2L_t+1)\times 1}$ is a vector of all zeros except a 1 at the $(L_t+1)$-th position.

\subsection{3-D STAP filter design in the correlation domain}
We first construct the virtual covariance matrix based on the virtual measurement model in (\ref{eq:VS2}). Considering that the clutter patch power \{$\sigma_{p,q}^2\}_{p=1,q=1}^{N_p,N_c}$ behaves like fully coherent sources, a 3-D spatial smoothing method is adopted to enhance the rank of virtual covariance matrix.

For $ 0\leq l_1 \leq L_s$, $0\leq l_2 \leq L_t$ and $ 0 \leq l_3 \leq L_s$, define a subvector of $\mathbf{z}$ as    \par\noindent
\begin{subequations}
\begin{align}            \label{eq:3-D_1}
\mathbf{z}_{l_1,l_2,l_3}  &= \sum_{p=1}^{N_p} \sum_{q=1}^{N_c} \sigma^{2}_{p,q}  \breve{\mathbf{v}}_{l_1,l_2,l_3}(\tilde{f}_T(p),f_R(\psi_{p,q}), f_d(\psi_{p,q}))  \notag \\
 &~~~~~ + \sigma^2_n \mathbf{e}_{l_1,l_2,l_3} 
\end{align}\normalsize
where 
\begin{align}
& \breve{\mathbf{v}}_{l_1,l_2,l_3}(\tilde{f}_T(p),f_R(\psi_{p,q}), f_d(\psi_{p,q}))  \nonumber \\
& ~~~ =\breve{\mathbf{a}}_{T,l_1}(\tilde{f}_T(p)) \otimes \breve{\mathbf{b}}_{l_2}(f_d(\psi_{p,q}))\otimes  \breve{\mathbf{a}}_{R,l_3}(f_R(\psi_{p,q}))
\end{align}
with
\begin{align}
&\breve{\mathbf{a}}_{T,l_1}(\tilde{f}_T(p))  \nonumber \\
&= [e^{-\mathrm{j}2\pi \tilde{f}_T(p)l_1},e^{\mathrm{j}2\pi \tilde{f}_T(p)(1-l_1)},\ldots,e^{\mathrm{j}2\pi f_T(p)(L_s-l_1)}]^T  \nonumber\\
& ~~~  \in \mathbb{C}^{(L_s+1)\times 1}, \\
&\breve{\mathbf{b}}_{l_2}(f_d(\psi_{p,q})) \nonumber\\ 
&= [e^{-\mathrm{j}2\pi f_d(\psi_{p,q})l_2},e^{\mathrm{j}2\pi f_d(\psi_{p,q}) (1-l_2)},\ldots,e^{\mathrm{j}2\pi f_d(\psi_{p,q})(L_t-l_2)}]^T   \nonumber\\
& ~~~ \in \mathbb{C}^{(L_t+1)\times 1}, \\
&\breve{\mathbf{a}}_{R,l_3}(f_R(\psi_{p,q})) \nonumber\\
&= [e^{-\mathrm{j}2\pi f_R(\psi_{p,q})l_3},e^{\mathrm{j}2\pi f_R(\psi_{p,q})(1-l_3)},\ldots,e^{\mathrm{j}2\pi f_R(\psi_{p,q})(L_s-l_3)}]^T \nonumber\\
& ~~~ \in \mathbb{C}^{(L_s+1)\times 1}
\end{align}
\end{subequations}
and $\mathbf{e}_{l_1,l_2,l_3} \triangleq \mathbf{e}_{T,l_1} \otimes \mathbf{e}_{d,l_2}\otimes \mathbf{e}_{R,l_3}$. $\mathbf{e}_{T,l_1} \in \mathbb{C}^{L_s \times 1}$ is a subvector constructed from the $(L_s+1-l_1)$-th entry to the $(2L_s+1-l_1)$-th entry of $\mathbf{e}_T$, $\mathbf{e}_{d,l_2} \in \mathbb{C}^{L_t \times 1}$ is a subvector constructed from the $(L_t+1-l_2)$-th entry to the $(2L_t+1-l_2)$-th entry of $\mathbf{e}_d$, $\mathbf{e}_{R,l_3} \in \mathbb{C}^{L_s \times 1}$ is a subvector constructed from the $(L_s+1-l_3)$-th entry to the $(2L_s+1-l_3)$-th entry of $\mathbf{e}_R$. So we have $\breve{\mathbf{a}}_{T,l_1}(\tilde{f}_T(p)) = \breve{\mathbf{a}}_{T,0}(\tilde{f}_T(p)) e^{-\mathrm{j}2\pi \tilde{f}_T(p)l_1}$, $\breve{\mathbf{b}}_{l_2}(f_d(\psi_{p,q}))= \breve{\mathbf{b}}_{0}(f_d(\psi_{p,q})) e^{-\mathrm{j}2\pi f_d(\psi_{p,q})l_2} $ and $\breve{\mathbf{a}}_{R,l_3}(f_R(\psi_{p,q})) = \breve{\mathbf{a}}_{R,0}(f_R(\psi_{p,q}))e^{-\mathrm{j}2\pi f_R(\psi_{p,q})l_3}$. Then, rewrite (\ref{eq:3-D_1}) as
\begin{align}             \label{eq:3-D_2}
\mathbf{z}_{l_1,l_2,l_3} &= \sum_{p=1}^{N_p} \sum_{q=1}^{N_c} \sigma^{2}_{p,q} e^{-\mathrm{j}2\pi (\tilde{f}_T(p)l_1 + f_d(\psi_{p,q})l_2+ f_R(\psi_{p,q})l_3)} \nonumber\\
&~~~ \times \breve{\mathbf{v}}_{0,0,0}(\tilde{f}_T(p),f_R(\psi_{p,q}), f_d(\psi_{p,q}))+ \sigma^2_n \mathbf{e}_{l_1,l_2,l_3}.
\end{align}

Thus, the spatial smoothing matrix $\mathbf{R}_v$ is computed as
\begin{align}           \label{eq:3-D_4}
\mathbf{R}_v = \sum_{l_1=0}^{L_s}\sum_{l_2=0}^{L_s} \sum_{l_3=0}^{L_t} \mathbf{z}_{l_1,l_2,l_3} \mathbf{z}^H_{l_1,l_2,l_3}.
\end{align}

Following Theorem~\ref{Theo_SSM} states the relationship between the spatial smoothing matrix $\mathbf{R}_v$ and space-time-range steering vectors.
\begin{theorem}      \label{Theo_SSM}
In terms of transmit steering subvector $\breve{\mathbf{a}}_{T,0}(\tilde{f}_T(p))$, receive steering subvector $ \breve{\mathbf{a}}_{R,0}(f_R(\psi_{p,q}))$, and time steering subvector $\breve{\mathbf{b}}_{0}(f_d(\psi_{p,q}))$, the spatial smoothing matrix $\mathbf{R}_v$ as defined in (\ref{eq:3-D_4}) is
\begin{align}               
\mathbf{R}_v &= C^2 \left( \sum_{p=1}^{N_p} \sum_{q=1}^{N_c} \sigma^2_{p,q}  \breve{\mathbf{v}}_{0,0,0}(\tilde{f}_T(p),f_d(\psi_{p,q}),f_R(\psi_{p,q})) \right. \nonumber\\ 
&~~~~~~~~~~~ \left. \breve{\mathbf{v}}_{0,0,0}^H(\tilde{f}_T(p),f_d(\psi_{p,q}),f_R(\psi_{p,q})) + \sigma_n^2 \mathbf{I}    \right)^2, \label{eq:3-D_5}
\end{align}
where $\breve{\mathbf{v}}_{0,0,0}(\tilde{f}_T(p),f_d(\psi_{p,q}),f_R(\psi_{p,q}))= \breve{\mathbf{a}}_{T,0}(\tilde{f}_T(p)) \otimes \breve{\mathbf{b}}_{0}(f_d(\psi_{p,q}))\otimes  \breve{\mathbf{a}}_{R,0}(f_R(\psi_{p,q}))$ and $C$ is a constant.
\end{theorem}
\begin{proof}
For notational simplicity, we replace the argument $\tilde{f}_T(p)$ by $p$ in boldface letters, e.g., $\breve{\mathbf{a}}_{T,0}(p) \triangleq \breve{\mathbf{a}}_{T,0}(\tilde{f}_T(p))$. Following (\ref{eq:3-D_2}), substituting $\mathbf{z}_{l_1,l_2,l_3}$ into (\ref{eq:3-D_4}) leads to the sum of three terms. The first term is
\begin{align}        \label{eq:3-D_7}
& \sum_{l_1,l_2,l_3} \sum_{p_1,q_1}\sum_{p_2,q_2} \sigma^2_{p_1,q_1} \sigma^2_{p_2,q_2}  e^{-j2\pi (\tilde{f}_T(p_1) - \tilde{f}_T(p_2))l_1}  \notag \\
& \times e^{-j2\pi (f_R(\psi_{p_1,q_1})-f_R(\psi_{p_2,q_2}))l_2}    e^{-j2\pi ( f_d(\psi_{p_1,q_1})-f_d(\psi_{p_2,q_2}))l_3}    \notag \\
&  \times \breve{\mathbf{v}}_{0,0,0}(p_1,q_1)\breve{\mathbf{v}}^H_{0,0,0}(p_2,q_2)         \notag \\
&= \sum_{p_1,q_1}\sum_{p_2,q_2}    \sigma^2_{p_1,q_1} \sigma^2_{p_2,q_2} \breve{\mathbf{v}}_{0,0,0}(p_1,q_1) \breve{\mathbf{v}}_{0,0,0}^H(p_1,q_1)  \notag \\
&~\times \breve{\mathbf{v}}_{0,0,0}(p_2,q_2) \breve{\mathbf{v}}^H_{0,0,0}(p_2,q_2)    \notag \\
&= \left( \sum_{p,q}\sigma^2_{p,q} \breve{\mathbf{v}}_{0,0,0}(p,q) \breve{\mathbf{v}}_{0,0,0}^H(p,q) \right)^2.
\end{align}
The second term is
\begin{align}                \label{eq:3-D_8}
& 2 \sum_{l_1,l_2,l_3} \sum_{p,q} \sigma^{2}_{p,q}\sigma^2_n e^{-j2\pi (\tilde{f}_T(p)l_1+ f_R(\psi_{p,q})l_2 + f_d(\psi_{p,q})l_3)}    \notag \\
&\times \breve{\mathbf{v}}_{0,0,0}(p,q)\mathbf{e}_{l_1,l_2,l_3}       \notag \\
&= 2 \sum_{p,q} \sigma^{2}_{p,q}\sigma^2_n \breve{\mathbf{v}}_{0,0,0}(p,q) \breve{\mathbf{v}}_{0,0,0}^H(p,q).
\end{align}
The third term becomes
\begin{align}                \label{eq:3-D_9}
\sum_{l_1,l_2,l_3} \sigma^4_n \mathbf{e}_{l_1,l_2,l_3}\mathbf{e}_{l_1,l_2,l_3}^H = \sigma^4_n  \mathbf{I}.
\end{align}
Combining (\ref{eq:3-D_7}), (\ref{eq:3-D_8}) and (\ref{eq:3-D_9}) concludes the proof.
\end{proof}

The result in Theorem~\ref{Theo_SSM} is similar to the two-dimensional (2-D) co-prime case stated in \cite{CLiu2015}, where the spatial smoothing matrix is related to space-time steering vectors. Following Theorem~\ref{Theo_SSM}, $\mathbf{R}_v$ in (\ref{eq:3-D_4}) is expressed as $\mathbf{R}_v = C^2 \widetilde{\mathbf{R}}_v^2$, where
\begin{align}                 \label{eq:3-D_9_11}
\widetilde{\mathbf{R}}_v = \sum_{p=1}^{N_p} \sum_{q=1}^{N_c} \sigma_{p,q}^2 \tilde{\mathbf{c}}_{p,q}\tilde{\mathbf{c}}^H_{p,q} + \sigma_n^2 \mathbf{I}   = \widetilde{\mathbf{R}}_c + \mathbf{R}_n,
\end{align}
with $\widetilde{\mathbf{R}}_c  = \sum_{p=1}^{N_p} \sum_{q=1}^{N_c} \sigma_{p,q}^2 \tilde{\mathbf{c}}_{p,q}\tilde{\mathbf{c}}^H_{p,q}$, $\mathbf{R}_n = \sigma_n^2 \mathbf{I}$ and $\tilde{\mathbf{c}}_{p,q}=     \breve{\mathbf{v}}_{0,0,0}(\tilde{f}_T(p),f_d(\psi_{p,q}),f_R(\psi_{p,q}))$.

Thus, the STAP filter weight vector that maximizes the output SINR in the correlation domain is
\begin{align}                 \label{eq:3-D_10}
\mathbf{w}_v = \frac{\widetilde{\mathbf{R}}_v ^{-1} \breve{\mathbf{v}}_t(p_0,\psi_0,\nu_0)}{\breve{\mathbf{v}}_t^H(p_0,\psi_0,\nu_0) \widetilde{\mathbf{R}}_v ^{-1}\breve{\mathbf{v}}_t(p_0,\psi_0,\nu_0)},
\end{align}
where $\breve{\mathbf{v}}_t(p_0,\psi_0,\nu_0) = \breve{\mathbf{a}}_{T,0}(\tilde{f}_T(p_0)) \otimes \breve{\mathbf{b}}_{0}(\nu_0))\otimes  \breve{\mathbf{a}}_{R,0}(f_R(\psi_{0}))$.

Co-array processing entails the merits of more DoFs. However, the size of the matrix $\widetilde{\mathbf{R}}_v $ is $(L_s+1)^2(L_t+1)\times (L_s+1)^2(L_t+1)$, which is significantly larger than its counterparts in the physical domain. Following (\ref{eq:3-D_10}), \textit{a fortiori} high computational complexity arises from the inversion of the large matrix $\widetilde{\mathbf{R}}_v$. To balance DoFs and computational complexity, we now propose a CoSTAP algorithm, wherein we approximate the clutter subspace in space-time-range domain by DPSSs.

\section{DPSS-Based Clutter Suppression}                      
\label{sec:DPSS}
Assume $M_T$ to be a positive integer and define $\Omega \triangleq \{0,1,\ldots,M_T-1\}$. For a given $\sigma$, referred to as the bandlimiting parameter, and each $i \in \Omega$, DPSSs \cite{DSlepian1978,DSlepian1983some} are defined as the discrete eigenfunctions of the eigenvalue problem given by 
\begin{align}                           \label{eq:DPSS_1}
\sum_{m=0}^{M_T-1} \frac{\sin 2\pi \sigma (n-m)}{\pi(n-m)}\bm{\phi}_i^{M_T,\sigma}[m] = \mu_i \bm{\phi}_i^{M_T,\sigma}[n],~~ n \in \Omega.
\end{align}
Equation (\ref{eq:DPSS_1}) is also valid for all integers $n \in \mathbb{Z}$ and it defines the values of the DPSSs on the real line. The $M_T$ eigenvalues are distinct, real and positive numbers indexed in the decreasing order of their magnitude, $1>\mu_0>\mu_1>\ldots>\mu_{M_T-1}>0$ and the corresponding discrete eigenfunctions, namely DPSSs, are denoted by $\bm{\phi}_i^{M_T,\sigma} \in \mathbb{C}^{M_T \times 1}, i=0,1,\ldots,M_T-1$, respectively. The superscripts of $\bm{\phi}_i^{M_T,\sigma}$ mean that the DPSSs are characterized by the $M_T$ and $\sigma$. The double orthogonality property of DPSSs is characterized by the following relations \cite{DSlepian1961prolate}:
\begin{align}
\begin{aligned}                         \label{eq:DPSS_2}
&\left\langle  \bm{\phi}_i^{M_T,\sigma}, \bm{\phi}_l^{M_T,\sigma}\right\rangle_{\Omega}=\sum_{n=0}^{M_T-1} \bm{\phi}_i^{M_T,\sigma}[n] \bm{\phi}_l^{M_T,\sigma}[n]= \delta_{i l}, \\
&\left\langle \bm{\phi}_i^{\infty,\sigma}, \bm{\phi}_l^{\infty,\sigma}\right\rangle_{\mathbb{Z}}=\sum_{n=-\infty}^{\infty} \bm{\phi}_i^{M_T,\sigma}[n] \bm{\phi}_l^{M_T,\sigma}[n]=\delta_{i l}/\mu_i,
\end{aligned}
\end{align}
where $\delta_{i l}$ denotes the Kronecker delta.

\subsection{Clutter rank analysis in the coarray domain}
For the sake of simplicity, denote $\bar{f} \triangleq f_R(\psi_{p,q})$ and $\bar{f}_p \triangleq \tilde{f}_T(p)= -2\Delta f r_u (p-1)/c$. Then, the Doppler frequency is $f_d(\psi_{p,q})= \frac{2\nu_{\mathrm{p}}T}{d} \bar{f}=\beta \bar{f}$, where $\beta \triangleq \frac{2 \nu_{\mathrm{p}}T}{d}$. Without loss of generality, we assume that $\beta$ is a positive rational number given by $\beta= \frac{M}{N}$ such that $M$ and $N$ are positive integers satisfying $\mathrm{gcd}(M,N)=1$. The corresponding compensated transmit steering vector, receive steering vector, and time steering vector of an arbitrary clutter patch are, respectively,
\begin{subequations}
\begin{align}           \label{eq:R_3-D_19}
\breve{\mathbf{a}}_{T,0}(\bar{f}_p )=  [1,e^{\mathrm{j}2\pi \bar{f}_p},\ldots,e^{\mathrm{j}2\pi \bar{f}_p L_s}]^T \in \mathbb{C}^{(L_s+1)\times 1},
\end{align}
\begin{align}
\breve{\mathbf{a}}_{R,0}(\bar{f} )=  [1,e^{\mathrm{j}2\pi \bar{f}},\ldots,e^{\mathrm{j}2\pi\bar{f}L_s}]^T  \in \mathbb{C}^{(L_s+1)\times 1},
\end{align}
and
\begin{align}
\breve{\mathbf{b}}_0(\bar{f}))  = [1,e^{\mathrm{j}2\pi \beta \bar{f}},\ldots,e^{\mathrm{j}2\pi \beta \bar{f}L_t}]^T \in \mathbb{C}^{(L_t+1)\times 1}.
\end{align}
\end{subequations}
Thus, for the $p$-th range region, the transmit-receive-time steering vector is
\begin{align}                  \label{eq:3-D_20}
\breve{\mathbf{v}}_{0,0,0}(\bar{f},\bar{f}_p)
&=\breve{\mathbf{a}}_{T,0}(\bar{f}_p) \otimes \breve{\mathbf{b}}_0(\bar{f})\otimes \breve{\mathbf{a}}_{R,0}(\bar{f} )             \notag \\
&= \breve{\mathbf{a}}_{T,0}(\bar{f}_p) \otimes [\mathbf{P}\mathbf{v}_{Rb}(\bar{f})]                                               \notag \\
&= [\breve{\mathbf{a}}_{T,0}(\bar{f}_p) \otimes \mathbf{P}] \mathbf{v}_{Rb}(\bar{f}),
\end{align}
where
\begin{align}                 \label{eq:3-D_20_P}
\mathbf{P}&=
\left[
\begin{matrix}
1 & & & &      \\
  & \ddots & & &   \\
  & & 1& &   \\
  &\underbrace{0~~~~~~~~~~~~1}_{\beta} & & &    \\
  &  & \ddots & &   \\
  &  &   &1 &   \\
  &  & & \ddots &    \\
0 & 0&\cdots&  &  1  \\
\end{matrix} \right]      \in \mathbb{C}^{(L_s+1)(L_t+1)\times R_r},
\end{align}
is a permutation matrix. In (\ref{eq:3-D_20_P}), the nonzero entries from the $[(l+1)(L_s+1)+1]$-th row to the $(l+2)(L_s+1)$-th row of $\mathbf{P}$ is the $\beta$ right-shifted version of those from the $[l(L_s+1)+1]$-th row to the $(l+1)(L_s+1)$-th row ,where $ 0\leq l \leq L_t $. When $\beta$ is a positive integer, all the entries of $\mathbf{P}$ can be considered to be located on the uniform two-dimensional grids with the grid resolution equal to 1 along the row (column) dimension. When $\beta = \frac{M}{N}$ is a positive fraction, the entries of $\mathbf{P}$ can be regarded to be located on denser uniform or nonuniform two-dimensional grids with the grid resolution equal to 1 along the row dimension and $\frac{1}{N}$ along the column dimension. The rank of $\mathbf{P}$ is $R_r$, namely the number of non-zero columns whose equivalent positions mapping to a filled ULA are denoted as $\bm{\zeta}=[\zeta_0,\zeta_1,\ldots,\zeta_{R_r-1}]\frac{2d}{\lambda_b}$ and $\zeta_0 =0$. Obviously, $R_r$ varies with $\beta$. Then, $\mathbf{v}_{Rb}(\bar{f})$ becomes
\begin{align}
\mathbf{v}_{Rb}(\bar{f}) &= [1,e^{\mathrm{j}2\pi \zeta_1\bar{f}},\ldots,e^{\mathrm{j}2\pi \zeta_{R_r-1}\bar{f}}]^T   \notag \\
                         &= [1,e^{\mathrm{j}2\pi \zeta_1 \frac{2d}{\lambda_b}\frac{\cos \psi}{2} },\ldots,e^{\mathrm{j}2\pi \zeta_{R_r-1} \frac{2d}{\lambda_b}\frac{\cos \psi}{2}}]^T   \notag \\
                         &~~ \in \mathbb{C}^{R_r \times 1}.
\end{align}

Following Theorem~\ref{theo_P} provides the closed-form of rank $R_r$ of $\mathbf{P}$.
\begin{theorem}                           \label{theo_P}
Define a permutation matrix $\mathbf{P} \in \mathbb{C}^{(L_s+1)(L_t+1)\times R_r}$ as displayed in (\ref{eq:3-D_20_P}) and a positive rational number $\beta=\frac{M}{N}$ where $\mathrm{gcd}(M,N)=1$. When $\frac{2d}{\lambda_b}=N$, the rank of $\mathbf{P}$ has the closed form
\begin{align}               
R_r&= \mathrm{rank}(\mathbf{P})            \notag \\
&= 
\begin{cases}
N(L_s+1)+M(L_t+1) \!\!-\!\! M N, ~ M\! <\! L_s \!+\!1 \\
                            \quad \quad \quad \quad \quad \quad\quad \quad \quad \quad \quad \quad \mathrm{and}~ N\!<\! L_t \!+\!1, \\
(L_s+1)(L_t+1),  ~ M \! \geq \! L_s \!+\!1 ~ \mathrm{or}~ N \! \geq \! L_t\!+\!1 .
\end{cases}           
\end{align}
\end{theorem}
\begin{proof}
First, consider the scenario when $\beta$ is a positive integer, namely $N=1$ and $\beta = M$. This case is similar to FDA-MIMO radar, where $\frac{2d}{\lambda_b}=1$ and the spacing ratio between transmit and receive array is assumed to be an integer \cite{KWang2022}. It can be readily shown that 
\begin{align}               \label{eq:3-D_20_3}
\mathrm{rank}(\mathbf{P})&= 
\begin{cases}
L_s+1+\beta L_t, &  \beta < L_s+1, \\
(L_s+1)(L_t+1),& \beta \geq L_s+1,
\end{cases}            \notag \\
&= \min\{L_s+1+M L_t, (L_s+1)(L_t+1)\}.
\end{align}

Next, consider when $\beta$ is not an integer. We first write the non-zero column positions of the $l$-th submatrix from the $[l(L_s+1)+1]$-th row to the $(l+1)(L_s+1)$-th row of $\mathbf{P}$ as 
\begin{align}
\bm{\zeta}^l&= \frac{2d}{\lambda_b}[\zeta^l_0,\zeta^l_0+1,\ldots,\zeta^l_0+ L_s]^T       \notag \\
&= \frac{2d}{\lambda_b}\left[\frac{lM}{N},\frac{lM+N}{N},\ldots,\frac{lM+L_sN}{N} \right]^T   \in \mathbb{C}^{(L_s+1)\times 1},
\end{align}
where $0\leq l \leq L_t$. Then, all non-zero column positions of $\mathbf{P}$ are
\begin{align}
 \bm{\Omega}_s = \frac{2d}{\lambda_b} [\bm{\zeta}^0,\bm{\zeta}^1,\ldots,\bm{\zeta}^{L_t}] \in  \mathbb{C}^{(L_s+1)\times (L_t+1)}.  
\end{align}
Obviously, the entries of $\bm{\Omega}_s$ may be redundant depending on $M$ and $N$. Our aim is to find out the number of distinct entries in $\bm{\Omega}_s$. When $\frac{2d}{\lambda_b}=N$, $\bm{\Omega}_s$ is
\begin{align}                \label{eq:3-D_20_5}
\bm{\Omega}_s    
= \left[
\begin{matrix}
0      & N        & \cdots &  L_s N       \\
M_c    & M+N    & \cdots &  M+L_s N   \\
\vdots & \vdots     & \cdots &  \vdots        \\
L_tM & L_tM+N & \cdots &  L_tM+L_sN
\end{matrix}
\right].
\end{align}
Following Lemma \ref{lemma_1} (see Appendix \ref{app:Intro_lemma1}), when $N \geq L_t+1$ or $M \geq L_s+1$, we have 
\begin{align}             \label{eq:3-D_20_7}
    \mathrm{rank}(\mathbf{P})=(L_s+1)(L_t+1).
\end{align}

Next, consider the case when both $N < L_t+1$ and $M < L_s+1$ are satisfied. From (\ref{eq:3-D_20_5}), the entries in $\bm{\Omega}_s$ range from $0$ to $L_tM+L_sN$ while there are holes in the range. Following Lemma \ref{lemma_2} (see Appendix \ref{app:Intro_lemma2}), which gives the number of holes under such circumstance, when $N < L_t+1$ and $M < L_s+1$, we have
\begin{align}                    \label{eq:3-D_20_9}
\mathrm{rank}(\mathbf{P})&= L_tM+L_sN+1-(M-1)(N-1)   \notag \\
 &= N(L_s+1)+M(L_t+1)-M N .
\end{align}
Then, invoking (\ref{eq:3-D_20_3}), (\ref{eq:3-D_20_7}) and (\ref{eq:3-D_20_9}) completes the proof. 
\end{proof}

Denote $\mathbf{P}(\bar{f}_p) \triangleq \breve{\mathbf{a}}_{T,0}(\bar{f}_p) \otimes \mathbf{P}$. Then, $\widetilde{\mathbf{R}}_c $ in (\ref{eq:3-D_9_11}) becomes
\begin{align}
\widetilde{\mathbf{R}}_c  &= \sum_{p=1}^{N_p} \sum_{q=1}^{N_c} \sigma_{p,q}^2 \breve{\mathbf{v}}_{0,0,0}(\bar{f},\bar{f}_p) \breve{\mathbf{v}}_{0,0,0}^H(\bar{f},\bar{f}_p)     \notag \\
&= \sum_{p=1}^{N_p} \sum_{q=1}^{N_c} \sigma_{p,q}^2 \mathbf{P}(\bar{f}_p)  \mathbf{v}_{Rb}(\bar{f}) \mathbf{v}_{Rb}^H(\bar{f}) \mathbf{P}^H(\bar{f}_p)                 \notag \\
&= [\mathbf{P}(\bar{f}_1),\mathbf{P}(\bar{f}_2),\ldots,\mathbf{P}(\bar{f}_{N_p})]  \mathbf{R}_b  [\mathbf{P}(\bar{f}_1),\mathbf{P}(\bar{f}_2),\ldots, \notag \\
&~~~~  \mathbf{P}(\bar{f}_{N_p})]^H
\end{align}
where the 
block diagonal matrix 
\begin{align}               \label{eq:3-D_17}
\mathbf{R}_b &= \mathrm{diag}\left(\left[\sum_{q=1}^{N_c} \sigma_{1,q}^2 \mathbf{v}_{Rb}(\bar{f}) \mathbf{v}_{Rb}^H(\bar{f}),\ldots, \right. \right.           \notag \\
& ~~~ \left.\left. \sum_{q=1}^{N_c} \sigma_{N_p,q}^2 \mathbf{v}_{Rb}(\bar{f}) \mathbf{v}_{Rb}^H(\bar{f})\right]\right).
\end{align}\normalsize
From the definition of a semipositive definite matrix, it follows from (\ref{eq:3-D_17}) that $\mathbf{R}_b$ is a semipositive definite matrix. Thus, the rank of clutter covariance matrix $\widetilde{\mathbf{R}}_c $ depends on the rank of $[\mathbf{P}(\bar{f}_1),\mathbf{P}(\bar{f}_2),\ldots,\mathbf{P}(\bar{f}_{N_p})]$, namely
\begin{align}                 \label{eq:3-D_31}
\mathrm{rank}(\widetilde{\mathbf{R}}_c) = \mathrm{rank}([\mathbf{P}(\bar{f}_1),\mathbf{P}(\bar{f}_2),\ldots,\mathbf{P}(\bar{f}_{N_p})]).
\end{align}

Observe that when $N_p=1$, the 3-D clutter rank is $\mathrm{rank}(\mathbf{P}(\bar{f}_1))= R_r$. When $N_p=2$, the clutter rank increases by $R_r$, which is equal to $\mathrm{rank}([\mathbf{P}(\bar{f}_1),\mathbf{P}(\bar{f}_2)])= 2R_r$. Following similar computations, when $\frac{2d}{\lambda_b}=N$, the clutter rank for an arbitrary $N_p$ is
\begin{align}               \label{eq:3-D_18}
\mathrm{rank}(\widetilde{\mathbf{R}}_c)=
\begin{cases}
N_p R_r, &  1\leq N_p \leq L_s+1, \\
(L_s+1)R_r,& N_p > L_s+1.
\end{cases}
\end{align}

\subsection{Clutter covariance approximation using DPSS}               
The exact closed-form expression (\ref{eq:3-D_18}) for 3-D clutter rank in the coarray domain is for the specific condition $\frac{2d}{\lambda_b}=N$. For the case $\frac{2d}{\lambda_b} \neq N$, we propose a clutter rank approximation method based on DPSS to estimate the clutter subspace when $1\leq N_p \leq L_s+1$.

From (\ref{eq:3-D_9_11}), we have $\mathrm{span}(\widetilde{\mathbf{R}}_c) = \mathrm{span}(\mathbf{C})$, where
\begin{align}
\mathbf{C}= [\tilde{\mathbf{c}}_{1,1},\ldots,\tilde{\mathbf{c}}_{N_p,N_c}] \in \mathbb{C}^{(L_s+1)^2(L_t+1)\times (N_p N_c)}.
\end{align}
The $(m,k,n)$-th, $0\leq m \leq L_s$, $0 \leq k \leq L_t$, $0 \leq n \leq L_s$, element of $\tilde{\mathbf{c}}_{p,q}$ is
\begin{align}            \label{eq:3-D_36}
\tilde{\mathbf{c}}_{p,q}[m,k,n] &=  e^{-\mathrm{j}2\pi (\tilde{f}_T(p)m + \beta \bar{f}k+ \bar{f}n)}         \notag \\
&=  e^{-\mathrm{j}2\pi\tilde{f}_T(p)m} e^{-\mathrm{j}2\pi \bar{f}(n + \beta k)}                    \notag \\
&= e^{-\mathrm{j}2\pi\tilde{f}_T(p)m}  g(n + \beta k;\bar{f}),
\end{align}
where $g(n + \beta k;\bar{f}) \triangleq e^{-\mathrm{j}2\pi \bar{f}(n + \beta k)}$. Evidently, $g(n + \beta k;\bar{f})$ may be viewed as a nonuniformly
sampled version of the truncated sinusoidal function
\begin{align}
g(t;\bar{f})=
\begin{cases}
e^{-\mathrm{j}2\pi\bar{f}t}, &0 \leq  t\leq L_s+\beta L_t,  \\
0,  & \text{otherwise}.
\end{cases}
\end{align}

Furthermore, $\bar{f} = d/\lambda_b \cos (\psi_{p,q})$ leads to $-d/\lambda_b \leq \bar{f} \leq d/\lambda_b$. Denote $T_b \triangleq L_s+\beta L_t$. Therefore, the energy of the signal $g(t;\bar{f})$ is largely confined to a certain time-frequency region $[0,T_b]\times [-d/\lambda_b,d/\lambda_b]$. Assume $\sigma=d/\lambda_b$. Such signals can be well approximated by linear combinations of $\lceil 2\sigma T_b+1 \rceil=\lceil 2dT_b/\lambda_b+1 \rceil$ orthogonal functions \cite{HJLandau1962}, namely the approximate clutter rank in space-time coarray is $\lceil 2dT_b/\lambda_b+1 \rceil$. The basis functions correspond to the continuous PSWF $\bm{\phi}_i(t), i= 0,1,\ldots,\lceil 2dT_b/\lambda_b \rceil$, whose discrete version are $\phi^{M_T,\sigma}_i[n]$ in (\ref{eq:DPSS_1}). Following the exact rank estimation result in Theorem \ref{theo_P}, set 
\begin{align}         \label{eq:R_3-D_NT}
    N_T= \begin{cases}
 R_r-1, &  2d/\lambda_b=N, \\
\lceil 2dT_b/\lambda_b\rceil,& 2d/\lambda_b \neq N.
\end{cases}
\end{align}

Thus, we have
\begin{align}         \label{eq:R_3-D_42}
g(n + \beta k;\bar{f}) &\approx \sum_{i=0}^{N_T} \alpha_{p,q,i} \phi_i(n + \beta k)        \notag \\
& = \sum_{i=0}^{N_T} \alpha_{p,q,i} \phi^{M_T,\sigma}_i[Nn + Mk],
\end{align}
where $\alpha_{p,q,i}$ is the weight coefficient for linear combination and $M_T = NL_s+ML_t+1$. Stacking the elements $\tilde{\mathbf{c}}_{p,q}[m,k,n]$ in (\ref{eq:3-D_36}) for different $m$, $k$ and $n$ yields
\begin{align}
\tilde{\mathbf{c}}_{p,q} & \approx \breve{\mathbf{a}}_{T,0}(\bar{f}_p ) \otimes  \sum_{i=0}^{N_T} \alpha_{p,q,i} \bm{\phi}_i         \notag \\
&= \sum_{i=0}^{N_T} \alpha_{p,q,i} \breve{\mathbf{a}}_{T,0}(\bar{f}_p ) \otimes  \bm{\phi}_i,
\end{align}
where $\bm{\phi}_i  \in \mathbb{C}^{(L_s+1)(L_t+1)\times 1}$ is a vector that consists of the elements $\phi^{M_T,\sigma}_i[Nn + Mk]$, namely the Slepian basis vector. So, we have
\begin{subequations}
\begin{align}                   \label{eq:3-D_43}
\mathrm{span}(\widetilde{\mathbf{R}}_c) = \mathrm{span}(\mathbf{C}) = \mathrm{span}(\breve{\mathbf{A}}_{T,0} \otimes \mathbf{U}_c),
\end{align}
where
\begin{align}
\mathbf{U}_c= [\bm{\phi}_0,\bm{\phi}_1,\ldots,\bm{\phi}_{N_T}] \in  \mathbb{C}^{(L_s+1)(L_t+1)\times (N_T+1)},
\end{align}
and
\begin{align}
\breve{\mathbf{A}}_{T,0}=[\breve{\mathbf{a}}_{T,0}(\bar{f}_1 ),\breve{\mathbf{a}}_{T,0}(\bar{f}_2 ),\ldots,\breve{\mathbf{a}}_{T,0}(\bar{f}_{N_p} )]   \in  \mathbb{C}^{(L_s+1) \times N_p}.
\end{align}
\end{subequations}
To sum up, the clutter subspace is estimated using the geometry of the signal rather than the data of received signal. As a result, this procedure is independent of data.

Denote $r_b \triangleq N_p(N_T+1)$ and $\mathbf{V}_c \triangleq \breve{\mathbf{A}}_{T,0} \otimes \mathbf{U}_c$. From (\ref{eq:3-D_43}), there exists a $r_b \times r_b$ matrix $\mathbf{D}_c$ such that 
\begin{align} \label{eq:R_3-D_59}
  \widetilde{\mathbf{R}}_c \approx \mathbf{V}_c \mathbf{D}_c \mathbf{V}_c^H,  
\end{align}
which means that $\widetilde{\mathbf{R}}_c $ is a low-rank clutter covariance matrix. Applying the matrix inversion lemma to (\ref{eq:3-D_9_11}) yields
\begin{align}       \label{eq:R_3-D_59_1}
\widetilde{\mathbf{R}}_v^{-1} = \mathbf{R}_n^{-1} - \mathbf{R}_n^{-1}\mathbf{V}_c (\mathbf{D}_c^{-1}+ \mathbf{V}_c^H \mathbf{R}_n^{-1} \mathbf{V}_c)^{-1} \mathbf{V}_c^H \mathbf{R}_n^{-1}.
\end{align}
In the above, $\mathbf{D}_c$ is estimated from (\ref{eq:R_3-D_59}) as
\begin{align}          \label{eq:3-D_60}
\widehat{\mathbf{D}}_c &= \mathbf{V}_c^{\dag} \widehat{\mathbf{R}}_c (\mathbf{V}_c^{\dag})^H    \notag \\
&= (\mathbf{V}_c^H \mathbf{V}_c)^{-1}\mathbf{V}_c^H (\widehat{\mathbf{R}}_v - \widehat{\mathbf{R}}_n) \mathbf{V}_c(\mathbf{V}_c^H \mathbf{V}_c)^{-1},
\end{align}
where $ \mathbf{V}_c^{\dag} \triangleq  (\mathbf{V}_c^H \mathbf{V}_c)^{-1}\mathbf{V}_c^H$ is the pseudo inverse of $\mathbf{V}_c$ and $\widehat{\mathbf{R}}_c \triangleq \widehat{\mathbf{R}}_v - \widehat{\mathbf{R}}_n $ is the estimated version of clutter covariance matrix $\widetilde{\mathbf{R}}_c$. The noise covariance matrix $\mathbf{R}_n$ may be estimated in the absence of clutter echoes by setting the radar transmitter to be in passive mode so that the receiver collects only white noise. 

In (\ref{eq:3-D_60}), the major computational complexity stems from matrix inversion. From the orthogonality property of DPSS in (\ref{eq:DPSS_2}), we have $\mathbf{U}_c^H \mathbf{U}_c=\mathbf{I}$. Then, the computational burden for computing $(\mathbf{V}_c^H \mathbf{V}_c)^{-1}\mathbf{V}_c^H$ is reduced via the following equivalence
\begin{align}
(\mathbf{V}_c^H \mathbf{V}_c)^{-1}\mathbf{V}_c^H &= [(\breve{\mathbf{A}}_{T,0} \otimes \mathbf{U}_c)^H (\breve{\mathbf{A}}_{T,0} \otimes \mathbf{U}_c)]^{-1} \mathbf{V}_c^H  \notag \\
&=  [(\breve{\mathbf{A}}_{T,0}^H \breve{\mathbf{A}}_{T,0}) \otimes (\mathbf{U}_c^H \mathbf{U}_c)]^{-1} \mathbf{V}_c^H   \notag \\
&= [(\breve{\mathbf{A}}_{T,0}^H \breve{\mathbf{A}}_{T,0})^{-1} \otimes \mathbf{I}] (\breve{\mathbf{A}}_{T,0} \otimes \mathbf{U}_c)^H \notag \\
&= (\breve{\mathbf{A}}_{T,0}^H \breve{\mathbf{A}}_{T,0})^{-1} \breve{\mathbf{A}}_{T,0}^H \otimes \mathbf{U}_c^H.
\end{align}

\subsection{Further interference rejection with inaccurate prior knowledge}
As stated in Section \ref{sec:DPSS}.B, we reconstruct the 3-D clutter covariance using DPSS without any interference. Then, the output SINR is
\begin{align}               \label{eq:SINR_pre}
    \mathrm{SINR}= \frac{\sigma_t^2 \vert \mathbf{w}_v^H \breve{\mathbf{v}}_t(p_0,\psi_0,\nu_0) \vert^2}{\mathbf{w}_v^H \widetilde{\mathbf{R}}_v \mathbf{w}_v^H },
\end{align}
where $\sigma_t^2$ is the power of the target. In some applications, unexpected clustered interference such as a barrage jamming signal is overlaid with the clutter signal leading to deterioration of SINR. Then, the approximate region of the interference is estimated \textit{a priori} and all possible contributions belonging to the corresponding region are removed from the original clutter covariance in $\widetilde{\mathbf{R}}_c$. 

From (\ref{eq:3-D_36}), rewrite the general form of the clutter vector as 
\begin{align}
    \tilde{\mathbf{c}}(\mathbf{f})= \breve{\mathbf{a}}_T(\tilde{f}_T) \otimes \breve{\mathbf{b}}(f_d)\otimes  \breve{\mathbf{a}}_R(f_R),
\end{align}
where $\mathbf{f}=[\tilde{f}_T,f_d,f_R]^T$ is a $3\times 1$ vector. Let $\mathbf{f}_0=[\tilde{f}_{T0},f_{d0},f_{R0}]^T$ be a given sample in the space-time-range 3-D region and $\mathcal{R}(\mathbf{f}_0,\bm{\Delta})= \{ \mathbf{f}\vert ~ \vert \tilde{f}_T - \tilde{f}_{T0} \vert <  \Delta_{f_T}/2, \vert f_d - f_{d0} \vert <  \Delta_{f_d}/2, f_R - f_{R0} \vert <  \Delta_{f_R}/2 \}$ be a given known region in which the interference should be rejected. Denote $\bm{\Delta }\triangleq [\Delta_{f_T},\Delta_{f_d},\Delta_{f_R}]^T$, where $ \Delta_{f_T}$ is the width of rejection interval along the transmitted frequency dimension centered on $\tilde{f}_{T0}$, $ \Delta_{f_d}$ is the width of rejection interval along the Doppler frequency dimension centered on $f_{d0}$, and $\Delta_{f_R}$ is the width of rejection interval along the received frequency dimension centered on $f_{R0}$. In order to perform the rejection in the continuous region $\mathcal{R}(\mathbf{f}_0,\bm{\Delta})$, we first need to approximate the region by a subspace and then eliminate the impact of the region using an orthogonal projection. 

Define the averaged projection residue over the region $\mathcal{R}(\mathbf{f}_0,\bm{\Delta})$ as 
\begin{align}              
\label{eq:Prior_3}
    \varepsilon (\mathbf{f}_0, \bm{\Delta}, \bm{\Pi}) &= \mathrm{E}_{\mathbf{f}} \left[ \left\Vert \bm{\Pi}^{\bot} \tilde{\mathbf{c}}(\mathbf{f})  \right\Vert^2   \right]          \notag  \\
    &=  \int\limits_{\mathcal{R}(\mathbf{f}_0,\bm{\Delta})}         \left\Vert \bm{\Pi}^{\bot} \tilde{\mathbf{c}}(\mathbf{f})  \right\Vert^2  p(f_T) p(f_d) p(f_R) \mathrm{d} \mathbf{f},
\end{align}
where $p(f_T)$, $p(f_d)$ and $p(f_R)$ are the probability density functions of the variables $f_T$, $f_d$ and $f_R$, respectively; $\bm{\Pi}$ is the orthogonal projector; and $\bm{\Pi}^{\bot} = \mathbf{I}- \bm{\Pi}$ is the orthogonal complement of $\bm{\Pi}$. The criterion of choosing the projector $\bm{\Pi}$ is to minimize the averaged projection residue. As a natural choice, we take the uniform distribution on $\mathcal{R}(\mathbf{f}_0,\bm{\Delta})$ into consideration. Therefore, (\ref{eq:Prior_3}) becomes
\begin{align}
    \varepsilon (\mathbf{f}_0, \bm{\Delta}, \bm{\Pi}) = \frac{1}{\Delta_{f_T} \Delta_{f_d} \Delta_{f_R}}   \int\limits_{\mathcal{R}(\mathbf{f}_0,\bm{\Delta})}   \left\Vert \bm{\Pi}^{\bot} \tilde{\mathbf{c}}(\mathbf{f})  \right\Vert^2 \mathrm{d} \mathbf{f}.
\end{align}

Recall the following result from \cite{JBosse2018} to construct the desired projector.
\begin{proposition}
\cite{JBosse2018} Let $\bm{\Pi}_{\widetilde{K}}$ denote an orthogonal projector of rank $\widetilde{K}$. Then, if we denote $\mathbf{v}_k$ as the $k$th eigenvector with corresponding eigenvalue $\lambda_k$ (ordered in decreasing order) of the matrix 
\begin{subequations}
\begin{align}            \label{eq:Prior_4}          
    \mathbf{R}_J = \mathrm{E}_{\mathbf{f}} \left[ \tilde{\mathbf{c}}(\mathbf{f}) \tilde{\mathbf{c}}^H (\mathbf{f})  \right],
\end{align}
we have 
\begin{align}
    \mathrm{arg}~\mathrm{min}_{\bm{\Pi}_{\widetilde{K}}} ~ \varepsilon (\mathbf{f}_0, \bm{\Delta}, \bm{\Pi}) = \sum_{k=1}^{\widetilde{K}} \mathbf{v}_k \mathbf{v}^H_k,    \label{eq:Prior_5} \\   
    \mathrm{min}_{\bm{\Pi}_{\widetilde{K}}} ~ \varepsilon (\mathbf{f}_0, \bm{\Delta}, \bm{\Pi}) = \sum_{k=\widetilde{K}+1}^{(L_s+1)^2(L_t+1)}  \lambda_k.
\end{align}
\end{subequations}
\end{proposition}
The criterion of choosing the total dimension $\widetilde{K}=\widetilde{K}_1 \widetilde{K}_2\widetilde{K}_3$ is also discussed in \cite{JBosse2018}, where $\widetilde{K}_1$, $\widetilde{K}_2$ and $\widetilde{K}_3$ are the dimensions along the transmit spatial, Doppler, and receive spatial domains, respectively. 

We now construct (\ref{eq:Prior_5}) from a computation-savvy perspective in our 3-D Kronecker case. Since $\mathbf{f}$ in $\mathcal{R}(\mathbf{f}_0,\bm{\Delta})$ obeys the uniform distribution, for the sake of simplicity, we choose $\Delta_{f_T}= \Delta_{f_R} =\frac{1}{L_s+1}$ and $\Delta_{f_d}=\frac{1}{L_t+1}$. Then, (\ref{eq:Prior_4}) is rewritten as (\ref{eq:Prior_7}), 
\begin{figure*}[!t]            
\par\noindent
\begin{align}    \label{eq:Prior_7}
    \mathbf{R}_J  &= \frac{1}{\Delta_{f_T} \Delta_{f_d} \Delta_{f_R}} \int\limits_{f_{T0}-\frac{\Delta_{f_T}}{2}}^{f_{T0}+\frac{\Delta_{f_T}}{2}}  \int\limits_{f_{d0}-\frac{\Delta_{f_d}}{2}}^{f_{d0}+\frac{\Delta_{f_d}}{2}}  \int\limits_{f_{R0}-\frac{\Delta_{f_R}}{2}}^{f_{R0}+\frac{\Delta_{f_R}}{2}}  
    [\breve{\mathbf{a}}_T(\tilde{f}_T) \otimes \breve{\mathbf{b}}(f_d)\otimes  \breve{\mathbf{a}}_R(f_R)] [\breve{\mathbf{a}}_T(\tilde{f}_T) \otimes \breve{\mathbf{b}}(f_d)\otimes  \breve{\mathbf{a}}_R(f_R)]^H   \mathrm{d}f_T \mathrm{d}f_d \mathrm{d}f_R             \notag \\
     &= \frac{1}{\Delta_{f_T} \Delta_{f_d} \Delta_{f_R}} \int\limits_{f_{T0}-\frac{\Delta_{f_T}}{2}}^{f_{T0}+\frac{\Delta_{f_T}}{2}}  \int\limits_{f_{d0}-\frac{\Delta_{f_d}}{2}}^{f_{d0}+\frac{\Delta_{f_d}}{2}}  \int\limits_{f_{R0}-\frac{\Delta_{f_R}}{2}}^{f_{R0}+\frac{\Delta_{f_R}}{2}}   
    [\breve{\mathbf{a}}_T(\tilde{f}_T)\breve{\mathbf{a}}_T^H(\tilde{f}_T)] \otimes [\breve{\mathbf{b}}(f_d)\breve{\mathbf{b}}^H(f_d) ]\otimes  [\breve{\mathbf{a}}_R(f_R)\breve{\mathbf{a}}_R^H(f_R)]   \mathrm{d}f_T \mathrm{d}f_d \mathrm{d}f_R             \notag \\
    &= \frac{1}{\Delta_{f_T} \Delta_{f_d} \Delta_{f_R}} \int\limits_{f_{T0}-\frac{\Delta_{f_T}}{2}}^{f_{T0}+\frac{\Delta_{f_T}}{2}}\breve{\mathbf{a}}_T(\tilde{f}_T)\breve{\mathbf{a}}_T^H(\tilde{f}_T) \mathrm{d}f_T           
    \otimes \int\limits_{f_{d0}-\frac{\Delta_{f_d}}{2}}^{f_{d0}+\frac{\Delta_{f_d}}{2}} \breve{\mathbf{b}}(f_d)\breve{\mathbf{b}}^H(f_d) \mathrm{d}f_d  \otimes  \int\limits_{f_{R0}-\frac{\Delta_{f_R}}{2}}^{f_{R0}+\frac{\Delta_{f_R}}{2}}  \breve{\mathbf{a}}_R(f_R)\breve{\mathbf{a}}_R^H(f_R) \mathrm{d}f_R   \notag \\
    &= \frac{1}{\Delta_{f_T}} \left\{ [\breve{\mathbf{a}}_T(\tilde{f}_{T0})\breve{\mathbf{a}}_T^H(\tilde{f}_{T0})] \diamond \mathbf{B}_{L_s+1,\Delta_{f_T}/2} \right\}        
    \otimes \frac{1}{\Delta_{f_d}} \left\{ [\breve{\mathbf{b}}(f_{d0})\breve{\mathbf{b}}^H(f_{d0})] \diamond \mathbf{B}_{L_t+1,\Delta_{f_d}/2} \right\}        \notag \\ 
    &~~~  \otimes  \frac{1}{\Delta_{f_R}} \left\{  [\breve{\mathbf{a}}_R(f_{R0})\breve{\mathbf{a}}_R^H(f_{R0})] \diamond \mathbf{B}_{L_s+1,\Delta_{f_R}/2} \right\}   \notag \\
    &= \frac{\mathbf{M}_T }{\Delta_{f_T}} \otimes  \frac{\mathbf{M}_d}{\Delta_{f_d}}  \otimes  \frac{\mathbf{M}_R}{\Delta_{f_R}},
\end{align}\normalsize
\hrulefill
\end{figure*}
where the matrix 
\begin{subequations}
\begin{align}
    \mathbf{B}_{P,W} = 2W \mathrm{sinc}(2W(k-l)), ~~~ 1\leq k,l\leq P,
\end{align}
and
\begin{align}
    \mathbf{M}_T &= [\breve{\mathbf{a}}_T(\tilde{f}_{T0})\breve{\mathbf{a}}_T^H(\tilde{f}_{T0})] \diamond \mathbf{B}_{L_s+1,\Delta_{f_T}/2} ,  \label{eq:Prior_9_1} \\
    \mathbf{M}_d &=  [\breve{\mathbf{b}}(f_{d0})\breve{\mathbf{b}}^H(f_{d0})] \diamond \mathbf{B}_{L_t+1,\Delta_{f_d}/2} ,  \label{eq:Prior_9_2}  \\
    \mathbf{M}_R &=  [\breve{\mathbf{a}}_R(f_{R0})\breve{\mathbf{a}}_R^H(f_{R0})] \diamond \mathbf{B}_{L_s+1,\Delta_{f_R}/2}.  \label{eq:Prior_9_3} 
\end{align}
\end{subequations}
Clearly, the matrices $\mathbf{M}_T$, $\mathbf{M}_d$ and $\mathbf{M}_R$ are ``modulated prolate matrices", i.e. prolate matrices modulated by $\tilde{f}_{T0}$, $f_{d0}$ and $f_{R0}$, respectively. For the prolate matrices $\mathbf{B}_{L_s+1,\Delta_{f_T}/2}$, $\mathbf{B}_{L_t+1,\Delta_{f_d}/2}$ and $ \mathbf{B}_{L_s+1,\Delta_{f_R}/2}$, singular value decomposition (SVD) yields, respectively,
\begin{subequations}
\begin{align}
   & \mathbf{B}_{L_s+1,\Delta_{f_T}/2}    \notag \\ 
   & ~~~~~ = \mathbf{U}(L_s+1,\Delta_{f_T}) \bm{\Lambda}(\lambda) \mathbf{U}^H(L_s+1,\Delta_{f_T}),    \label{eq:Prior_10}   \\
   & \mathbf{B}_{L_t+1,\Delta_{f_d}/2}    \notag \\
   & ~~~~~ = \mathbf{U}(L_t+1,\Delta_{f_d}) \bm{\Lambda}(\mu) \mathbf{U}^H(L_t+1,\Delta_{f_d}) ,       \label{eq:Prior_11}   
\end{align}
and
\begin{align}
   & \mathbf{B}_{L_s+1,\Delta_{f_R}/2}    \notag \\
   & ~~~~~ = \mathbf{U}(L_s+1,\Delta_{f_R}) \bm{\Lambda}(\kappa) \mathbf{U}^H(L_s+1,\Delta_{f_R}) ,     \label{eq:Prior_12}
\end{align}
\end{subequations}
where the $k_1$-th column of $\mathbf{U}(L_s+1,\Delta_{f_T})$ is the $k_1$-th DPSS characterized by timelimiting parameter $L_s+1$ and bandlimiting parameter $\Delta_{f_T}$; the $k_2$-th column of $\mathbf{U}(L_T+1,\Delta_{f_d})$ is the $k_2$-th DPSS characterized by $L_t+1$ and $\Delta_{f_d}$; the $k_3$-th column of $\mathbf{U}(L_s+1,\Delta_{f_R})$ is the $k_3$-th DPSS characterized by $L_s+1$ and $\Delta_{f_R}$; $\bm{\Lambda}(\lambda)$, $\bm{\Lambda}(\mu)$, and $\bm{\Lambda}(\kappa)$ are all diagonal matrices with the diagonal elements $[\lambda_1,\lambda_2,\ldots,\lambda_{L_s+1}]$, $[\mu_1,\mu_2,\ldots,\mu_{L_t+1}]$, and $[\kappa_1,\kappa_2,\ldots,\kappa_{L_s+1}]$, respectively. 

The singular vectors of $\mathbf{M}_T$, $\mathbf{M}_d$, and $\mathbf{M}_R$  are the ``modulated DPSS" vectors \cite{BPDay2020}, where the singular vectors of $\mathbf{M}_T$ are $\{\bm{\alpha}_{k_1}=\breve{\mathbf{a}}_T(\tilde{f}_{T0}) \diamond \mathbf{u}_{k_1}(L_s+1,\Delta_{f_T})\}_{k_1=1}^{L_s+1} $; the singular vectors of $\mathbf{M}_d$ are $\bm{\beta}_{k_2}= \{\breve{\mathbf{b}}(\tilde{f}_{d0}) \diamond \mathbf{u}_{k_2}(L_t+1,\Delta_{f_d})\}_{k_2=1}^{L_t+1}$; and the singular vectors of $\mathbf{M}_R$ are $ \{\bm{\gamma}_{k_3} =\breve{\mathbf{a}}_R(\tilde{f}_{R0}) \diamond \mathbf{u}_{k_3}(L_s+1,\Delta_{f_R})\}_{k_3=1}^{L_s+1}$. The following Theorem~\ref{Theorem:eigen} provides eigenvectors for a 3-D Kronecker structure.
\begin{theorem}                \label{Theorem:eigen}
Assume the eigenvectors of modulated prolate matrices $\mathbf{M}_T \in \mathbb{C}^{L_s+1}$, $\mathbf{M}_d \in \mathbb{C}^{L_t+1}$, and $\mathbf{M}_R \in \mathbb{C}^{L_s+1}$ are $\{\bm{\alpha}_{k_1} \}_{k_1=1}^{L_s+1}$, $\{\bm{\beta}_{k_2} \}_{k_2=1}^{L_t+1}$, and $\{\bm{\gamma}_{k_1} \}_{k_1=1}^{L_s+1}$ and the corresponding eigenvalues are $\{\lambda_{k_1} \}_{k_1=1}^{L_s+1}$, $\{\mu_{k_2} \}_{k_2=1}^{L_t+1}$, and $\{\kappa_{k_3} \}_{k_3=1}^{L_s+1}$, respectively. Then, the $k$-th, $1\leq k \leq (L_s+1)^2(L_t+1)$, eigenvector of $\mathbf{M}_T \otimes \mathbf{M}_d \otimes \mathbf{M}_R$ is
\begin{align}                      \label{eq:DPSS_3-D}
 \mathbf{v}_k= \bm{\alpha}_{k_1} \otimes \bm{\beta}_{k_2} \otimes \bm{\gamma}_{k_3}.   
\end{align}
\end{theorem}
\begin{proof}
Following (\ref{eq:Prior_9_1})$\sim$(\ref{eq:Prior_9_3}) and (\ref{eq:Prior_10})$\sim$(\ref{eq:Prior_12}), the modulated prolate matrices are 
\par\noindent
\begin{align}
    \mathbf{M}_T & =  \mathrm{diag}(\breve{\mathbf{a}}_T(\tilde{f}_{T0}))\mathbf{U}(L_s+1,\Delta_{f_T}) \bm{\Lambda}(\lambda)   \notag \\
       &~~~~  \times \mathbf{U}^H(L_s+1,\Delta_{f_T}) \mathrm{diag}(\breve{\mathbf{a}}_T(\tilde{f}_{T0}))^H       \notag  \\
        &= \mathbf{T} \bm{\Lambda}(\lambda)  \mathbf{T}^H,                               \\
    \mathbf{M}_d &= \mathrm{diag}(\breve{\mathbf{b}}(\tilde{f}_{d0}))   \mathbf{U}(L_t+1,\Delta_{f_d}) \bm{\Lambda}(\mu)   \notag \\      
        &~~~~ \times \mathbf{U}^H(L_t+1,\Delta_{f_d})\mathrm{diag}(\breve{\mathbf{b}}(\tilde{f}_{d0}))^H                 \notag     \\
        &= \mathbf{P} \bm{\Lambda}(\mu)  \mathbf{P}^H,                                \\
    \mathbf{M}_R &= \mathrm{diag}(\breve{\mathbf{a}}_R(\tilde{f}_{R0})) \mathbf{U}(L_s+1,\Delta_{f_R}) \bm{\Lambda}(\kappa)    \notag \\  
        & ~~~~  \times \mathbf{U}^H(L_s+1,\Delta_{f_R})     \mathrm{diag}(\breve{\mathbf{a}}_R(\tilde{f}_{R0}))^H        \notag \\
        &= \mathbf{Q} \bm{\Lambda}(\kappa)  \mathbf{Q}^H,    
\end{align} \normalsize
where $\mathbf{T}=  \mathrm{diag}(\breve{\mathbf{a}}_T(\tilde{f}_{T0}))  \mathbf{U}(L_s+1,\Delta_{f_T})= [\bm{\alpha}_1,\bm{\alpha}_2,\ldots,\bm{\alpha}_{L_s+1}]$, $\mathbf{P}= \mathrm{diag}(\breve{\mathbf{b}}(\tilde{f}_{d0})) \mathbf{U}(L_t+1,\Delta_{f_d})= [\bm{\beta}_1,\bm{\beta}_2,\ldots,\bm{\beta}_{L_t+1}]$ and $\mathbf{Q}= \mathrm{diag}(\breve{\mathbf{a}}_R(\tilde{f}_{R0})) \mathbf{U}(L_s+1,\Delta_{f_R})=[\bm{\gamma}_1,\bm{\gamma}_2,\ldots,\bm{\gamma}_{L_s+1}]$. Then, 
\begin{align}
    & \mathbf{M}_T \otimes \mathbf{M}_d \otimes \mathbf{M}_R         \notag \\
    &~~ = (\mathbf{T} \bm{\Lambda}(\lambda)  \mathbf{T}^H )\otimes (\mathbf{P} \bm{\Lambda}(\mu)  \mathbf{P}^H)  \otimes (\mathbf{Q} \bm{\Lambda}(\kappa)  \mathbf{Q}^H)                      \notag \\
    &~~ =(\mathbf{T}\otimes \mathbf{P} \otimes \mathbf{Q})( \bm{\Lambda}(\lambda) \otimes \bm{\Lambda}(\mu) \otimes \bm{\Lambda}(\kappa)) (\mathbf{T}\otimes \mathbf{P} \otimes \mathbf{Q})^H.
\end{align}
Clearly, the $k$-th column of $\mathbf{V} \triangleq \mathbf{T}\otimes \mathbf{P} \otimes \mathbf{Q}$ expressed as $ \mathbf{v}_k= \bm{\alpha}_{k_1} \otimes \bm{\beta}_{k_2} \otimes \bm{\gamma}_{k_3} $ is the eigenvector of $\mathbf{M}_T \otimes \mathbf{M}_d \otimes \mathbf{M}_R$. This concludes the proof.
\end{proof}

Therefore, from (\ref{eq:Prior_5}) and (\ref{eq:DPSS_3-D}), we obtain the orthogonal projector $\bm{\Pi}_{\widetilde{K}}$. Then, the jamming-free clutter covariance matrix is $\breve{\mathbf{R}}_v =\bm{\Pi}_{\widetilde{K}}^{\bot} \widetilde{\mathbf{R}}_v $, which counteracts the effect of the unexpected interference region. The output SINR after clutter subspace rejection becomes
\begin{align}           \label{eq:SINR_post}
    \mathrm{SINR}_{\mathrm{post}}= \frac{\sigma_t^2 \vert \mathbf{w}_v^H \breve{\mathbf{v}}_t(p_0,\psi_0,\nu_0) \vert^2}{\mathbf{w}_v^H \bm{\Pi}_{\widetilde{K}}^{\bot} \widetilde{\mathbf{R}}_v\mathbf{w}_v^H }.
\end{align}
Clearly, by comparing the pre- and post-rejection SINRs in (\ref{eq:SINR_pre}) and (\ref{eq:SINR_post}), respectively, it is evident that removing clustered interference enhances the output SINR. This improves the probability of target detection, which is a monotonically increasing function of SINR \cite{SSK1994}.

Algorithm~\ref{alg:DPSS_Scheme} summarizes our proposed DPSS-based clutter suppression method for co-pulsing FDA radar.
\begin{algorithm}
	\caption{DPSS-based co-pulsing FDA clutter suppression}
	\label{alg:DPSS_Scheme}
	\begin{algorithmic}[1]
		\Statex \textbf{Input:} Sample covariance matrix $\widetilde{\mathbf{R}}_v$ in the space-time-range coarray domain (\ref{eq:3-D_9_11})
		\Statex \textbf{Output:} Approximated $\widetilde{\mathbf{R}}_v^{-1}$ and CoSTAP filter weight $\mathbf{w}_v$
        
        \Statex
		\Statex $\triangleright$ \textbf{Clutter subspace construction (off-line)}
		\State Compute the exact clutter rank based on Theorem \ref{theo_P} as $R_r$ and the approximate rank in \cite{HJLandau1962} as $\lceil 2dT_b/\lambda_b+1 \rceil$.
		
		\State Calculate $N_T$ based on (\ref{eq:R_3-D_NT}). 
		
		\State Generate the Slepian basis vectors $\{\bm{\phi}_i \}_{i=0}^{N_T} \in \mathbb{C}^{(L_s+1)(L_t+1)\times 1}$.
        \Statex $\mathbf{U}_c \leftarrow [\bm{\phi}_0,\bm{\phi}_1,\ldots,\bm{\phi}_{N_T}]$, $\mathbf{V}_c \leftarrow \breve{\mathbf{A}}_{T,0} \otimes \mathbf{U}_c$.
        
        \Statex
		\Statex $\triangleright$ \textbf{Unexpected interference region rejection}
		\State Compute $\mathbf{v}_k$ from the modulated DPSS vectors $\bm{\alpha}_{k_1}$, $ \bm{\beta}_{k_2}$, $\bm{\gamma}_{k_3}$ based on Theorem \ref{Theorem:eigen}.
        \State Construct the orthogonal projector $\bm{\Pi}_{\widetilde{K}}$ from (\ref{eq:Prior_5}).
        \State $\breve{\mathbf{R}}_v \leftarrow \bm{\Pi}_{\widetilde{K}}^{\bot} \widetilde{\mathbf{R}}_v $.
  
		\Statex
		\Statex $\triangleright$ \textbf{CoSTAP filter}
		\State $\widehat{\mathbf{D}}_c \leftarrow (\mathbf{V}_c^H \mathbf{V}_c)^{-1}\mathbf{V}_c^H (\breve{\mathbf{R}}_v - \widehat{\mathbf{R}}_n) \mathbf{V}_c(\mathbf{V}_c^H \mathbf{V}_c)^{-1}$.
		
		\State $\widetilde{\mathbf{R}}_v^{-1} \leftarrow \widehat{\mathbf{R}}_n^{-1} - \widehat{\mathbf{R}}_n^{-1}\mathbf{V}_c (\widehat{\mathbf{D}}_c^{-1}+ \mathbf{V}_c^H \widehat{\mathbf{R}}_n^{-1} \mathbf{V}_c)^{-1} \mathbf{V}_c^H \widehat{\mathbf{R}}_n^{-1}$.

         \State Compute $\mathbf{w}_v$ based on (\ref{eq:3-D_10}).
	\end{algorithmic}
\end{algorithm}
\color{black}

\subsection{Computational complexity analysis}  
From (\ref{eq:3-D_10}), the computational complexity of our proposed CoSTAP method primarily depends on the complexity of computing $\widehat{\mathbf{R}}$ in (\ref{eq:R_PA3}), $\mathbf{R}_v$ in (\ref{eq:3-D_4}), and the inversion of $\widetilde{\mathbf{R}}_v$. The complexities of computing $\widehat{\mathbf{R}}$ and $\mathbf{R}_v$ are of the order $\mathcal{C}_1=\mathcal{O}(L(P_s^2K)^2)$ and $\mathcal{C}_2=\mathcal{O}((L_s+1)^6(L_t+1)^3)$, respectively. If we invert the covariance matrix $\widetilde{\mathbf{R}}_v$ directly, its computational complexity is $\mathcal{C}_{c1}=\mathcal{O}((L_s+1)^6(L_t+1)^3)$, leading to the overall complexity $\mathcal{C}_{1}+\mathcal{C}_{2}+\mathcal{C}_{c1}$. 

Leveraging upon the geometry structure of the received signal in coarray domain, our proposed DPSS-based clutter suppression method approximates the inversion of $\widetilde{\mathbf{R}}_v$ from a set of Slepian basis vectors generated offline. From (\ref{eq:R_3-D_59_1}), the complexity of computing $\widetilde{\mathbf{R}}_v^{-1}$ is $\mathcal{C}_{c2}=\mathcal{O}(2r_b^3+r_b(L_s+1)^4(L_t+1)^2)$, where $\mathcal{O}(2r_b^3)$ is the complexity of computing $\mathbf{D}_c^{-1}$ and $(\mathbf{D}_c^{-1}+ \mathbf{V}_c^H \mathbf{R}_n^{-1} \mathbf{V}_c)^{-1}$. Thus, the overall complexity of our proposed method is $\mathcal{C}_{1}+\mathcal{C}_{2}+\mathcal{C}_{c2}$. 

For a comparison, we consider the uniform counterpart which has the same number of physical sensors and pulses. From (\ref{eq:R_PA3}), the complexity of computing the inversion of $\widehat{\mathbf{R}}$ directly is $\mathcal{C}_{p}=\mathcal{O}((P_s^2K)^3)$, resulting in the overall complexity of $\mathcal{C}_1+\mathcal{C}_{p}$. Our proposed CoSTAP method requires more computational operations than the classical uniform FDA-STAP because the former operates on the virtual signal in coarray domain. However, CoSTAP provides a better output SINR performance. Moreover, the DPSS-based clutter suppression method allows striking a balance between the performance and computational complexity. We demonstrate this via numerical experiments in the next section.

\color{black}

\section{Numerical Experiments}                     
\label{sec:NE}
We performed various simulation experiments to validate the theoretical derivation and demonstrate the effectiveness of our proposed method. For the airborne FDA radar, the co-prime pattern $\mathcal{S}_s$ of co-prime arrays and co-prime FOs was characterized by $M_s=2$ and $N_s=3$. Thus, the total number of array sensor elements is $P_s= N_s+2M_s-1=6$. For co-prime PRI, the co-prime co-prime pattern $\mathcal{S}_t$ was characterized by $M_t=2$ and $N_t=3$. Hence, a total of $K=N_t+2M_t-1=6$ pulses were transmitted during the CPI with the fundamental PRI and pulse duration of $T=0.5$ ms and $T_p=1$ $\mu$s, respectively. The reference carrier frequency was set as $f_b=1$ GHz. For the setting of frequency increment $\Delta f$, to separate the range ambiguous clutter in the transmit and receive spatial domain effectively, we adopted the criterion in \cite{KWang2022}, namely $\Delta f \equiv 1/N_p \mod 1/T$, where $N_p$ is the number of clutter ambiguities. The clutter-to-noise ratio was set to 40 dB. All clutter scatterers followed the assumption on cosine values of conic angles as stated in (\ref{eq:SM_R3}).


\noindent\textbf{Clutter power spectrum analysis}: Assume that the height and velocity of the platform are $H = 6$ km and $ v_{\mathrm{p}}= 150$ m/s, respectively. When $d=\lambda_b/2$, we have $\beta = \frac{2 v_{\mathrm{p}}T}{d}=1$. The clutter power spectrum is drawn with the MVDR algorithm as
\begin{align}
& P_{MVDR}(f_T,f_d,f_R)  \notag \\
&~~~~~ =\frac{1}{\breve{\mathbf{v}}_{0,0,0}^H(f_T,f_d,f_R)\widetilde{\mathbf{R}}_c^{-1} \breve{\mathbf{v}}_{0,0,0}(f_T,f_d,f_R)},
\end{align}
where $\breve{\mathbf{v}}_{0,0,0}(f_T,f_d,f_R)=\breve{\mathbf{a}}_{T,0}(f_T) \otimes \breve{\mathbf{b}}_{0}(f_d)\otimes  \breve{\mathbf{a}}_{R,0}(f_R)$ is the 3-D steering vector and $(f_T,f_d,f_R)$ is the 3-D grid satisfying $f_T \in [0,1]$, $f_d \in [-0.5,0.5]$ and $f_R \in [-0.5,0.5]$. We first considered $N_p=3$.  Fig. \ref{Clutter_spectra_Np3} shows the clutter spectra of the traditional airborne FDA and our proposed co-pulsing airborne FDA. In our proposed method, we consider two covariance matrix estimation methods in the coarray domain: one based on Theorem \ref{Theo_SSM}, where $\widetilde{\mathbf{R}}_v$ is estimated using finite secondary training samples, and the other based on the co-pulsing DPSS approximation method. It can be observed that all three methods successfully separate the clutter spectra from different ambiguous range regions in the transmit-receive domain. Moreover, our proposed co-pulsing DPSS-based method effectively approximates the clutter subspace.

However, when the number of ambiguities increases to $N_p=6$, as shown in Fig. \ref{Clutter_spectra_Np6}, the traditional airborne FDA fails to separate the clutter spectra from different ambiguous range regions in the transmit-receive domain, whereas our proposed co-pulsing airborne FDA continues to work effectively. This is because our methods focus on the virtual domain processing provided by the C-Cube architecture, whose DoFs are not restricted by the number of physical FOs. Specifically, as seen in (\ref{eq:R_3-D_19}), the maximum resolvable number of range ambiguities for our proposed co-pulsing FDA is $L_s=7$, determined by the number of FOs in the coarray domain. In contrast, the range ambiguity resolvability of the classical uniform FDA is $P_s-1=5$, limited by the number of FOs in the physical domain.
\begin{figure*}[htp]
\centerline{\includegraphics[width=0.95\textwidth]{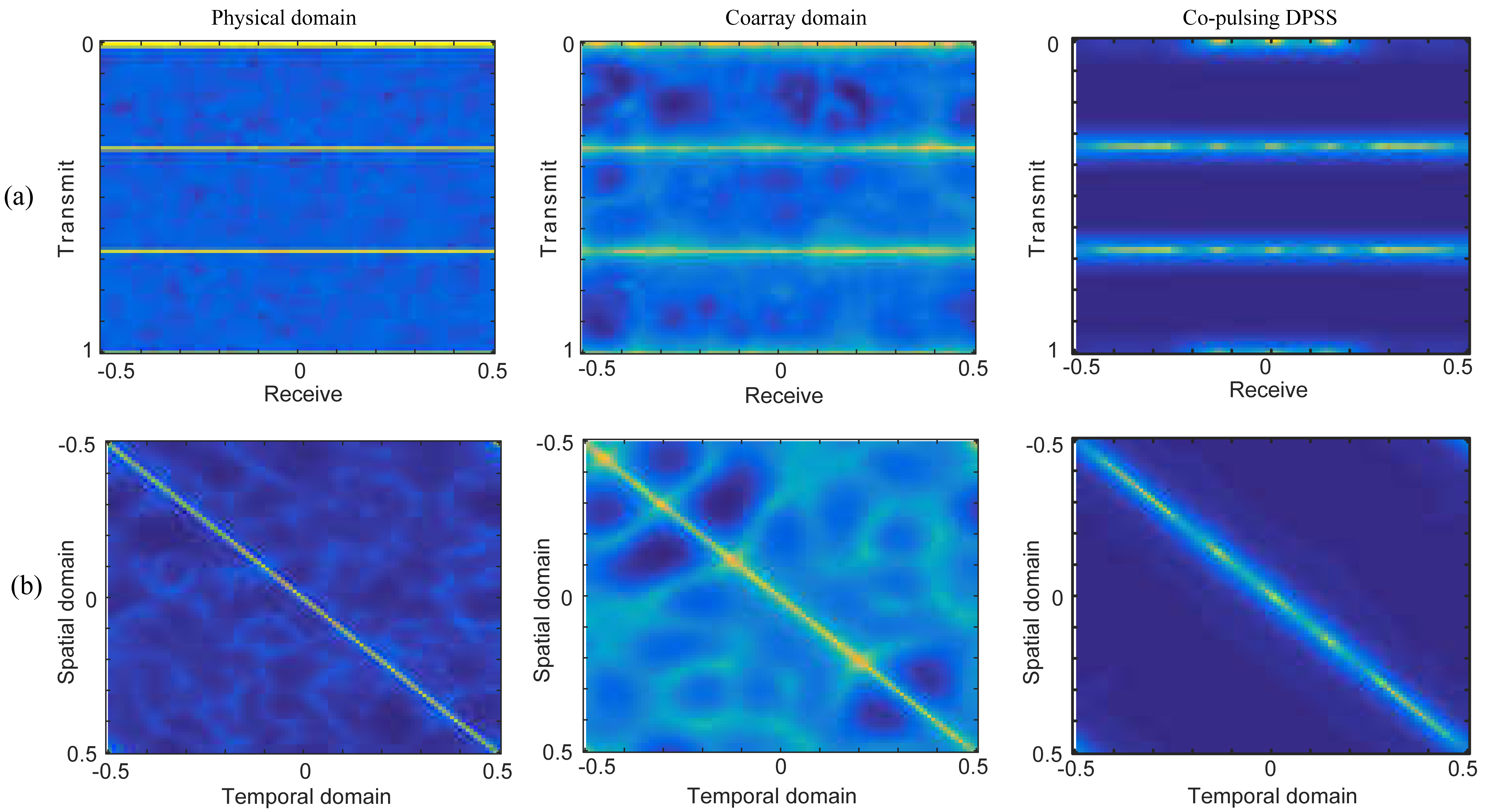}}
\caption{Clutter spectra of physical domain (left column), coarray domain (middle column), and co-pulsing DPSS (right column), respectively when $N_p=3$. Spectra in (a) transmit-receive domains and (b) spatial and temporal domains.
}
\label{Clutter_spectra_Np3}
\end{figure*}
\begin{figure*}[htp]
\centerline{\includegraphics[width=0.95\textwidth]{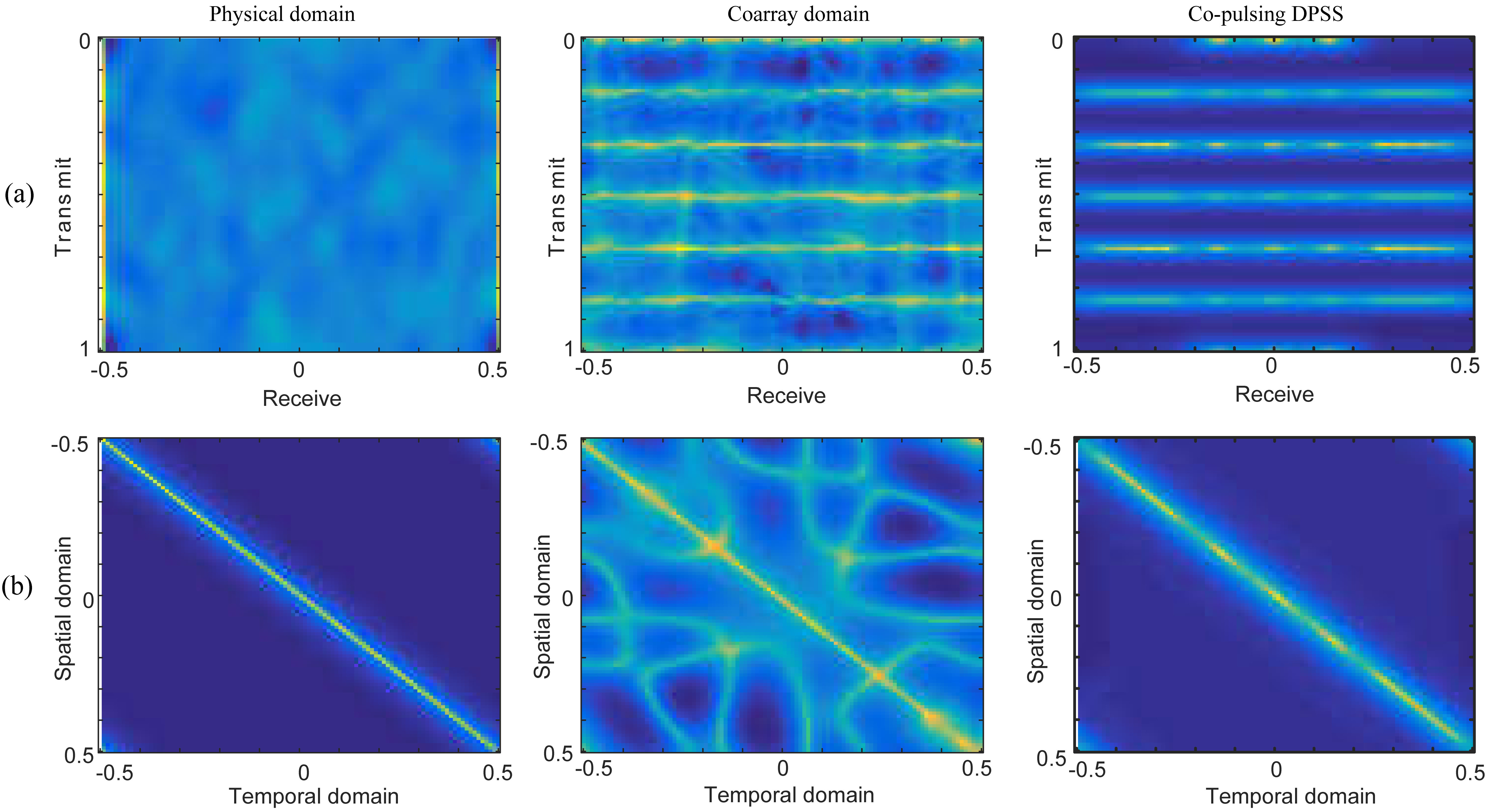}}
\caption{As in Fig. \ref{Clutter_spectra_Np3}, but for $N_p=6$.
}
\label{Clutter_spectra_Np6}
\end{figure*}

\begin{figure*}[ht]
\centerline{\includegraphics[width=0.95\textwidth]{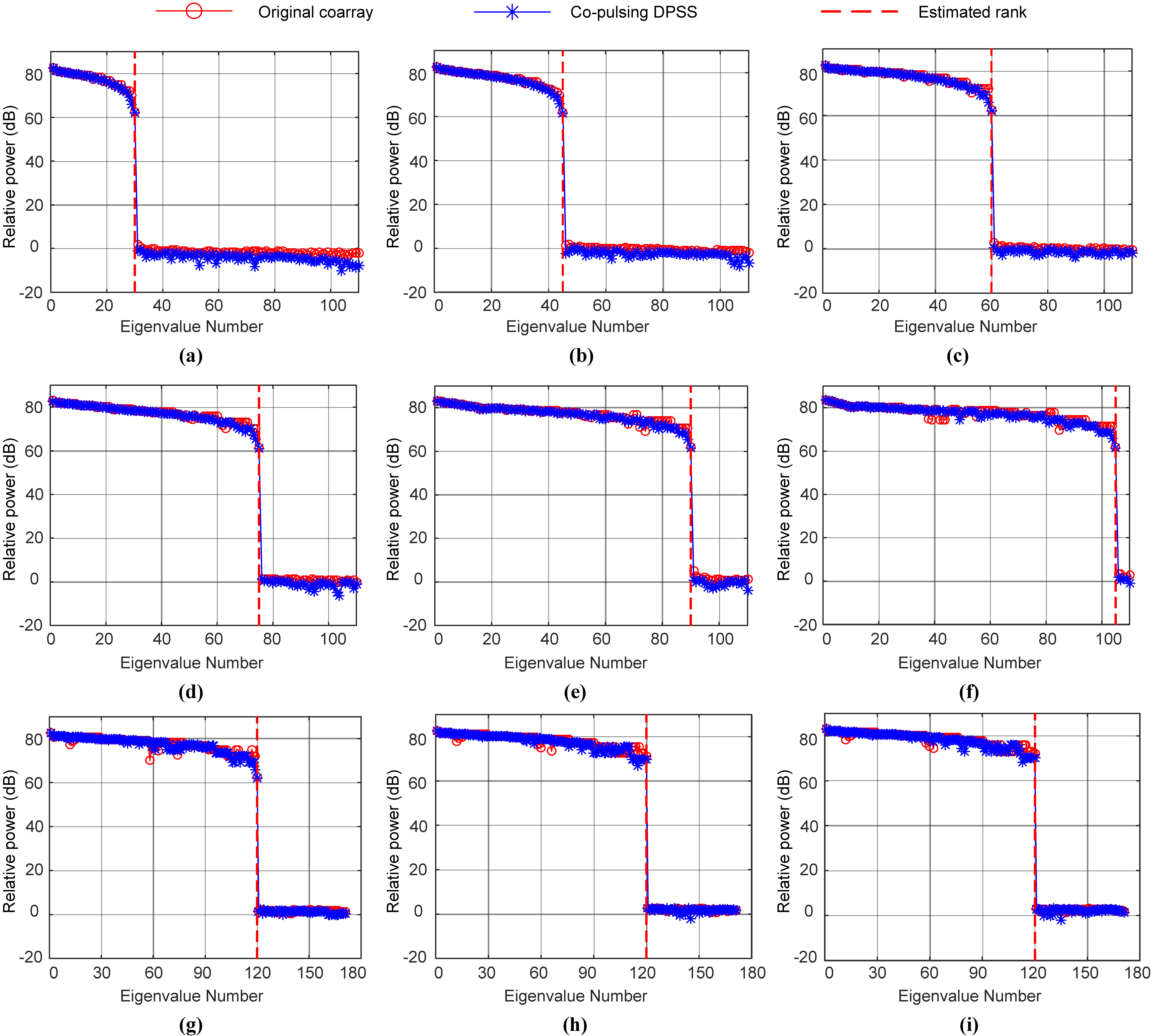}}
\caption{Clutter rank estimation results with different range ambiguities for $\beta=1$. (a) $N_p=2$. (b) $N_p=3$. (c) $N_p=4$. (d) $N_p=5$. (e) $N_p=6$. (f) $N_p=7$. (g) $N_p=8$. (h) $N_p=9$. (i) $N_p=10$.   
}
\label{Clutter_rank}
\end{figure*}

\begin{figure*}[ht]
\centerline{\includegraphics[width=0.95\textwidth]{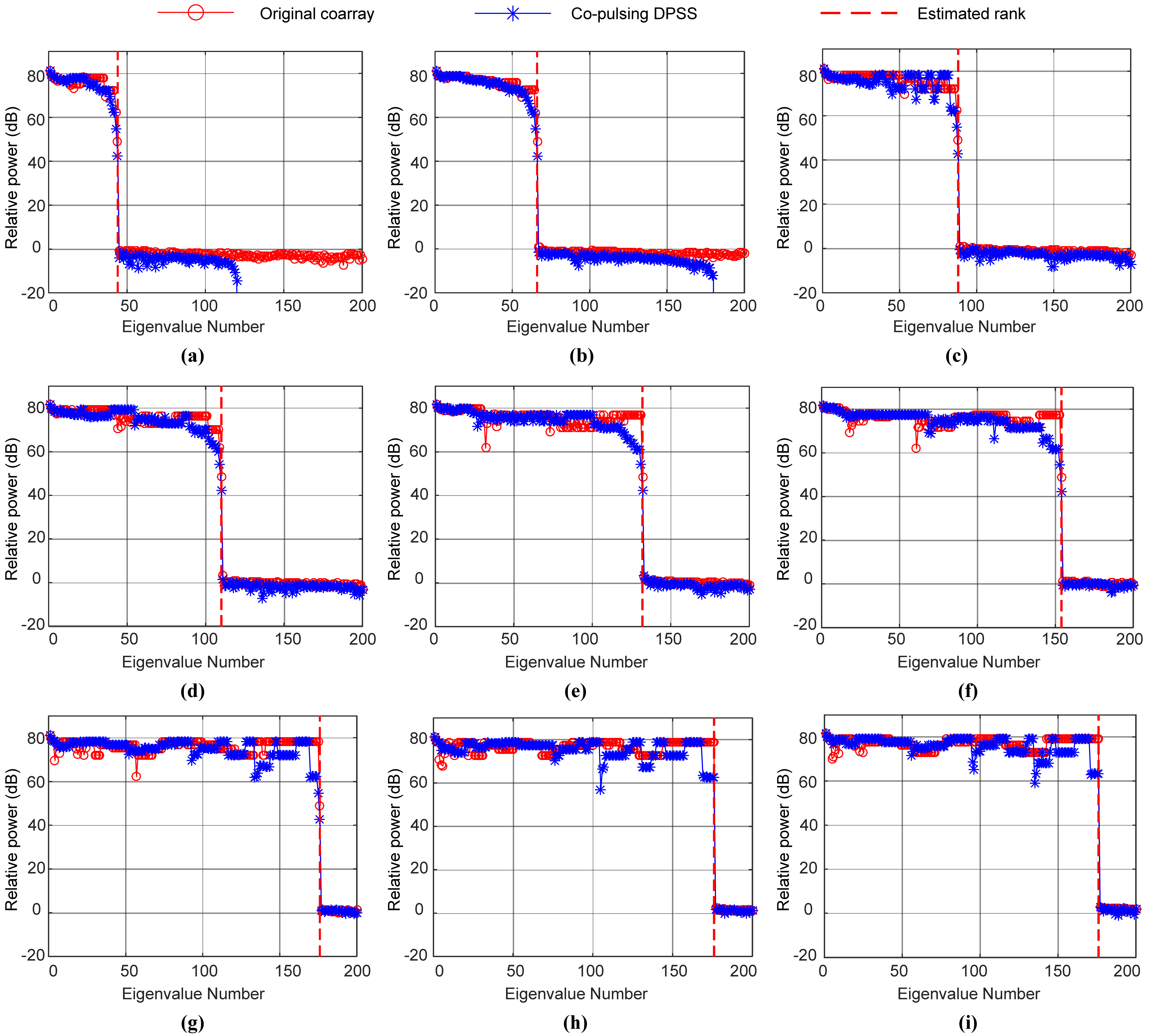}}
\caption{As in Fig.~\ref{Clutter_rank} but for $\beta=0.5$. 
}
\label{Clutter_rank_beta_0.5}
\end{figure*}

\noindent\textbf{Clutter rank estimation}: Next, we assessed the accuracy of the clutter rank prediction. Figures \ref{Clutter_rank} and \ref{Clutter_rank_beta_0.5} illustrate the clutter rank estimation results with different range ambiguities (marked with dotted red line) for $ v_{\mathrm{p}}= 150$ m/s ($\beta =1$) and $ v_{\mathrm{p}}= 75$ m/s ($\beta =0.5$), respectively. For comparison, we also plot the actual eigenvalue power of the coarray covariance matrix and co-pulsing DPSS-based approximated covariance matrix. The clutter rank values with different methods corresponding to Figures \ref{Clutter_rank} and \ref{Clutter_rank_beta_0.5} are displayed in Table \ref{tbl:Clutter_rank_value} to clearly show the accuracy of our proposed clutter rank estimation method. It follows from Figures \ref{Clutter_rank}-\ref{Clutter_rank_beta_0.5} and Table \ref{tbl:Clutter_rank_value} that our proposed rank prediction criteria in (\ref{eq:3-D_18}) accurately estimates clutter rank for different range ambiguities.  When $N_p \leq 8$, the clutter rank is proportional to the number of range ambiguities. However, when $N_p>8$, the clutter rank is a constant because the maximum resolvable number of range ambiguities in our proposed CoSTAP is restricted by $L_s+1$. The eigenvalue power curves of the original coarray covariance matrix closely follow that of the co-pulsing DPSS-based approximated version indicating that our proposed co-pulsing DPSS-based method accurately approximates the 3-D clutter subspace. 

    \begin{table*}
    \caption{The Estimated Clutter Rank Values for Different Methods}
    \label{tbl:Clutter_rank_value}       
    \centering
    \begin{threeparttable}
    \begin{tabular}{P{1cm}P{2.3cm}P{1cm}P{1cm}P{1cm}P{1cm}P{1cm}P{1cm}P{1cm}P{1cm}P{1cm}}
    \hline\noalign{\smallskip}
    ~   & $N_p$ & 2 & 3 & 4 & 5 & 6 & 7 & 8 & 9 & 10
    \\
    \noalign{\smallskip}
    \hline
    \noalign{\smallskip}
                & Original coarray  & 30 & 45 & 60 & 75 & 90 & 105 & 120 & 120 & 120 \\
                 \noalign{\smallskip}
    $\beta = 1$ & Co-pulsing DPSS   & 30 & 45 & 60 & 75 & 90 & 105 & 120 & 120 & 120 \\
     \noalign{\smallskip}
                & Estimated rank (\ref{eq:3-D_18})    & 30 & 45 & 60 & 75 & 90 & 105 & 120 & 120 & 120 \\
                 \noalign{\smallskip}
    \hline
    \noalign{\smallskip}
                & Original coarray  & 44 & 66 & 88 & 110 & 132 & 154 & 176 & 176 & 176 \\
                 \noalign{\smallskip}
 $\beta = 0.5$ & Co-pulsing DPSS    & 44 & 66 & 88 & 110 & 132 & 154 & 176 & 176 & 176 \\
                  \noalign{\smallskip}
                & Estimated rank (\ref{eq:3-D_18})    & 44 & 66 & 88 & 110 & 132 & 154 & 176 & 176 & 176 \\
    \noalign{\smallskip}\hline\noalign{\smallskip}
    \end{tabular}
    \end{threeparttable}
    \end{table*}

\noindent\textbf{Performance of CoSTAP}: We tested the SINR performance of different STAP methods through a full-dimensional (FD) processing. We compared the uniform FD STAP in physical domain, standard co-pulsing FD STAP, and DPSS-based co-pulsing FD STAP methods in coarray domain. We adopted the SMI-MVDR algorithm as a benchmark here because it provides good performance when the number of i.i.d. training samples are large. We fixed the number of i.i.d. training samples to 500. Fig. \ref{STAP_performance} shows that both of our proposed co-pulsing FD STAP methods achieve better performance with at least $4$ dB higher gain in the output SINR than uniform FD STAP. This is because co-pulsing FDA radar entails more virtual sensor elements and pulses in coarray domain than those of uniform FDA in physical domain. In addition, comparing the DPSS-based co-pulsing FD STAP and normal co-pulsing FD STAP shows that the former exhibits about $8$ dB SINR loss even though it enjoys the advantage of low complexity. Fig. \ref{runtime_comp} illustrates the computational complexity of all three STAP methods by plotting the resultant runtime versus number of physical sensors. It follows from Figures \ref{STAP_performance}-\ref{runtime_comp} that our proposed DPSS-based STAP method provides a balance between the output SINR performance and computational complexity.

\begin{figure}[t]
\centerline{\includegraphics[width=0.48\textwidth]{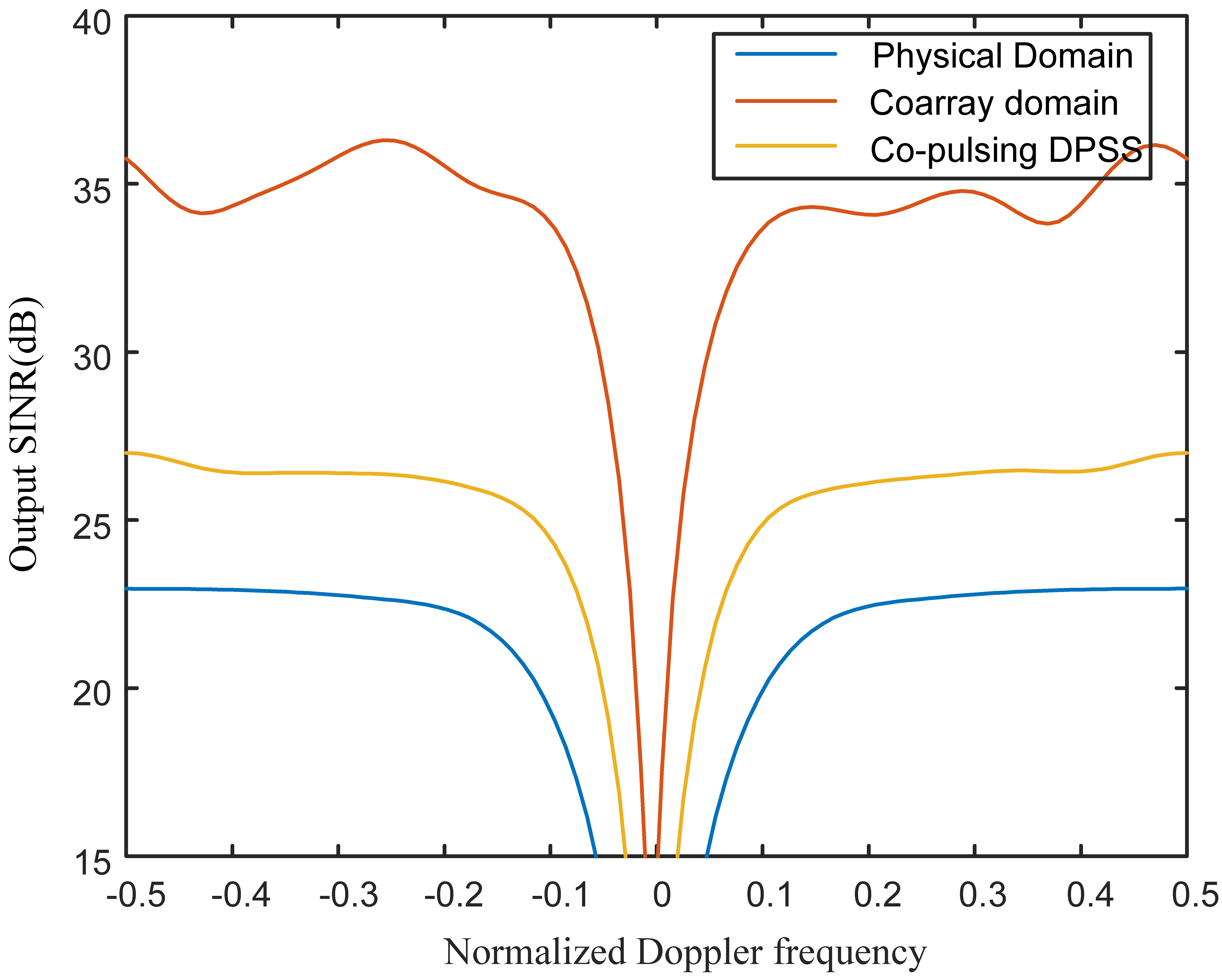}}
\caption{SINR comparison across different STAP methods.  
}
\label{STAP_performance}
\end{figure}

\begin{figure}[t]
\centerline{\includegraphics[width=0.47\textwidth]{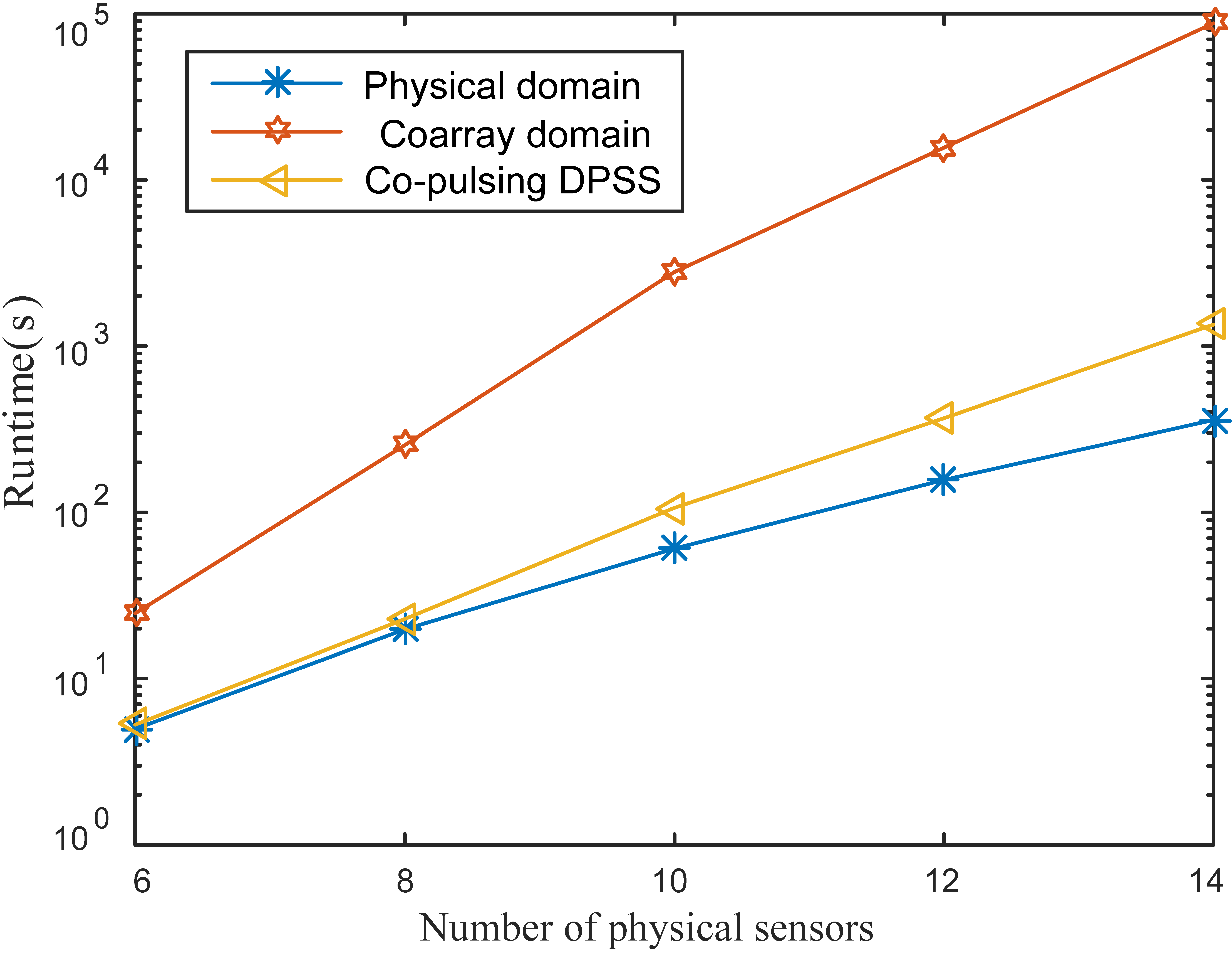}}
\caption{Runtime for different STAP methods with respect to the number of physical array elements.  }
\label{runtime_comp}
\end{figure}

\noindent\textbf{Robustness to interference}: Assume that unexpected interferences are located in the continuous region $\mathcal{R}(\mathbf{f}_0,\bm{\Delta})$. Here, we set $\mathbf{f}_0=[0.1,-0.3,0.3]$ and $\bm{\Delta}=[\frac{1}{L_s+1},\frac{1}{L_t+1},\frac{1}{L_s+1}]=[1/7,1/7,1/7]$. The range, Doppler, and angle are uniformly distributed on the region $\mathcal{R}(\mathbf{f}_0,\bm{\Delta})$ for each of the interference components. The interference-to-noise ratio was set at 30 dB implying that the interference is much stronger than clutter. Fig. \ref{Robustness_to_interference}(a) depicts the clutter and interference spectra prior to any rejection. We observe that strong interference exists in $\mathcal{R}(\mathbf{f}_0,\bm{\Delta})$ which may affect clutter subspace. Fig. \ref{Robustness_to_interference}(b) depicts the spectra produced after applying the interference rejection algorithm over the highlighted interference region. Clearly, our proposed 3-D subspace rejection method effectively suppresses the undesired interference in the 3-D region. Fig. \ref{STAP_rejection} exhibits the marked improvement of SINR before and after clutter subspace rejection, especially in the clustered interference region.
\begin{figure}[t]
\centerline{\includegraphics[width=0.5\textwidth]{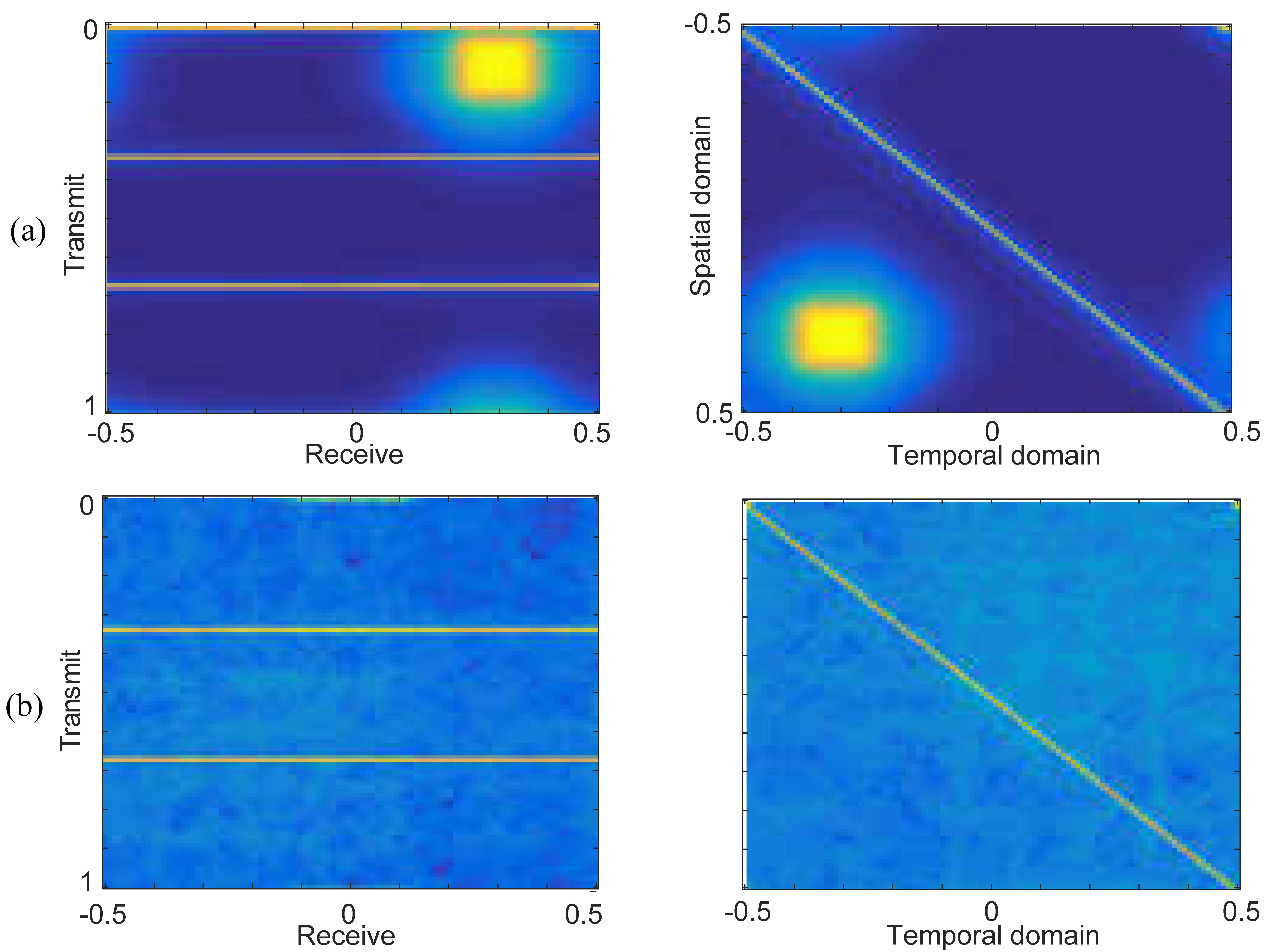}}
\caption{Clutter and interference spectra in the coarray domain (a) before and (b) after applying the interference rejection algorithm over the yellow highlighted interference region.
}
\label{Robustness_to_interference}
\end{figure}

\begin{figure}[t]
\centerline{\includegraphics[width=0.48\textwidth]{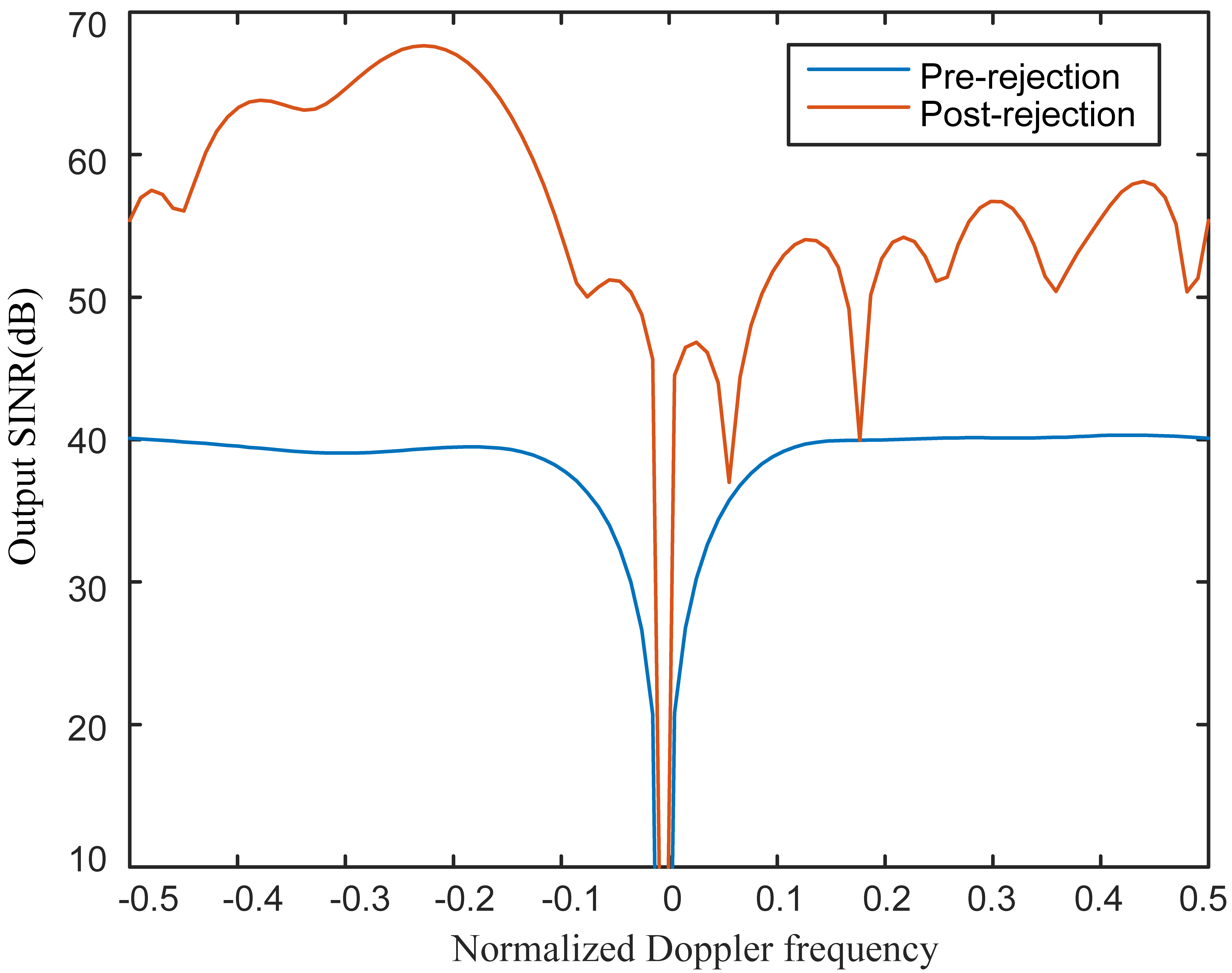}}
\caption{Output SINR (a) before and (b) after clutter subspace rejection.
}
\label{STAP_rejection}
\end{figure}


\section{Summary}                     
\label{sec:Conclusions}
We explored a novel STRAP filter for clutter suppression based on co-pulsing FDA radar, which offers the advantage of a larger aperture and increased DoFs. This leads to improved range ambiguity resolvability and better output SINR performance compared to uniform counterparts. Considering that the increased DoFs in the space-time-range correlation domain result in a significantly higher computational burden, we incorporated a DPSS-based subspace approximation scheme to achieve a low-complexity adaptive processing algorithm. To accurately determine the size of the 3-D clutter subspace, we developed a robust clutter rank estimation approach that provides a closed-form solution. Simulation results demonstrated a good estimate of clutter rank and showed that our proposed co-pulsing methods offer enhanced 3-D clutter characterization capability and superior clutter suppression performance compared to the ULA counterpart in the physical domain. 

While our approach focused on a linear co-prime FDA, this work may be extended to CoSTAP for 2-D co-prime FDA, such as L-shaped co-prime FDA \cite{WLv2022}, where more DoFs can be utilized. DPSS-based approach may also be utilized for the standard FDA-MIMO STAP \cite{JXu2017}, which is a special case of our setup as stated in Theorem~\ref{theo_P}. 

\appendices
\section{}
\label{app:Intro_lemma1}
\begin{lemma}                   \label{lemma_1}
If $N \geq L_t+1$ or $M \geq L_s+1$, the number of distinct entries of the matrix $\bm{\Omega}_s$ defined in (\ref{eq:3-D_20_5}) is $(L_s+1)(L_t+1)$, namely all the entries of $\bm{\Omega}_s$ are distinct.
\end{lemma}
\begin{proof}
We provide the proof by assuming the contrary. Consider the circumstance that $N \geq L_t+1$. Suppose the $(l_1,p_1)$-th entry is the same as the $(l_2,p_2)$-th entry,  where $0\leq l_1 < l_2 \leq L_t < N$ and $0\leq p_2 < p_1 \leq L_s$. Thus, we have $l_1M+p_1N=l_2M+p_2N$, that is $M/N = (p_1-p_2)/(l_2-l_1)$. However, since $l_2-l_1 < N$,  this contradicts co-primality of $M$ and $N$. The proof for the case $M \geq L_s+1$ follows \textit{mutatis mutandis}. This completes the proof. 
\end{proof}

\section{}
\label{app:Intro_lemma2}
\begin{lemma}                   \label{lemma_2}
If $N < L_t+1$ and $M < L_s+1$, the number of holes in $\bm{\Omega}_s$ is $(M-1)(N-1)$.
\end{lemma}
\begin{proof}
Define $k \triangleq lM+pN$ where $l$ and $p$ are integers. Recall the property of co-prime pairs of integers in \cite{SQin2017_2}: With $l$ and $p$ restricted to $l\in [0,N-1]$ and $p\in [0,M-1]$, the integer $k  \in [0,M(N-1)+N(M-1)]$ has $MN$ distinct values. Then, the number of holes $N_{\mathrm{hole}}$ is 
\begin{align}
N_{\mathrm{hole}} &= M(N-1)+N(M-1)+1 - MN    \notag \\
&=(M-1)(N-1).
\end{align}
When $l\in [0,L_t]$ and $p\in [0,L_s]$ where $N < L_t+1$ and $M < L_s+1$, the values of $k$ are restricted to the range $[0,L_tM+L_sM]$ and the contiguous part is $[(M-1)(N-1),L_sN+L_tM-(M-1)(N-1)]$. Obviously, as the $L_s$ and $L_t$ increase, the right edge of the contiguous part also increases while the left edge is fixed. So, the hole positions below $(M-1)(N-1)$ are also fixed which are the same as those in case when $L_s=M-1$, $L_t=N-1$, whose number of holes is $(M-1)(N-1)/2$. Considering the hole positions in the range $[0,L_tM+L_sM]$ are symmetric, the total number of holes are $(M-1)(N-1)$. This completes the proof. 
\end{proof}

\bibliographystyle{IEEEtran}
\bibliography{main}

\end{document}